\documentclass[11pt,a4paper]{article}
\pdfoutput=1

\usepackage[font=small]{caption}
\usepackage{tikz}
\usetikzlibrary{plotmarks,calc,decorations, decorations.pathmorphing, patterns}
\usetikzlibrary{arrows}

\tikzstyle dynkin node=[very thick,shape=circle,draw,inner sep=0pt,minimum size=5mm]
\tikzstyle dynkin line=[very thick]
\tikzstyle inverse line=[gray,line width=1.46pt,line cap=round, dash pattern=on 0pt off 2\pgflinewidth]
\tikzstyle red phase=[red,decoration={snake,amplitude=0.1mm,segment length=1.6mm},decorate]
\tikzstyle blue phase=[blue,decoration={snake,amplitude=0.1mm,segment length=0.9mm},decorate]
\tikzstyle green phase=[green,decoration={snake,amplitude=0.1mm,segment length=0.9mm},decorate]
\tikzstyle brown phase=[brown,decoration={snake,amplitude=0.1mm,segment length=0.9mm},decorate]

\tikzstyle arrow=[thick,rounded corners=18pt,-latex]
\tikzstyle box=[draw,rounded corners,outer sep=4pt]
\tikzstyle B node=[outer sep=0pt]
\tikzstyle Q node=[inner sep=1pt,outer sep=0pt]

\usepackage[nosort]{cite}

\usepackage[british]{babel} 
\usepackage{amsthm}
 
\usepackage{amsmath,amssymb,amsthm}
\usepackage{mathrsfs}
\usepackage{bbm}
\usepackage{bm}
\usepackage{slashed}
\usepackage{braket}
\usepackage{dsfont}

\usepackage{color}
\usepackage[textwidth = 430 pt, textheight = 630 pt]{geometry}
\definecolor{MyDarkBlue}{rgb}{0.15,0.25,0.45}

\usepackage[linktocpage=true,hypertexnames=false]{hyperref}
\hypersetup{
colorlinks=true,
citecolor=MyDarkBlue,
linkcolor=MyDarkBlue,
urlcolor=MyDarkBlue,
pdfauthor={Andrew Rolph${}^a$ and Alessandro Torrielli${}^b$},
pdftitle={Secret Symmetries of Type IIB Superstring Theory on AdS3 x S3 x M4},
pdfsubject={hep-th}
breaklinks=true
} 

\newcommand{\Ya}[1]{\mathcal{Y}(#1)}

\let\fn\footnote
\renewcommand{\footnote}[1]{\linespread{1.1}\fn{#1}\linespread{1.29}}

\makeatletter\renewcommand{\section}{\@startsection
{section}{1}{\z@}{-3.5ex plus -1ex minus
    -.2ex}{2.3ex plus .2ex}{\bf\mathversion{bold} }}
\makeatletter\renewcommand{\subsection}{\@startsection{subsection}{2}{\z@}{-3.25ex
plus -1ex minus
   -.2ex}{1.5ex plus .2ex}{\bf\mathversion{bold} }}
\makeatletter\renewcommand{\subsubsection}{\@startsection{subsubsection}{3}{-2.45ex}{-3.25ex
plus -1ex minus -.2ex}{1.5ex plus .2ex}{\it }}
\renewcommand{\thesection}{\arabic{section}}
\renewcommand{\thesubsection}{\arabic{section}.\arabic{subsection}}
\renewcommand{\@seccntformat}[1]{\@nameuse{the#1}.~~}

\renewcommand{\theequation}{\thesection.\arabic{equation}}
\makeatletter \@addtoreset{equation}{section}

\renewcommand*\l@section{\@dottedtocline{1}{0em}{2em}}
\renewcommand*\l@subsection{\@dottedtocline{2}{2em}{2.4em}}
\renewcommand*\l@subsubsection{\@dottedtocline{4}{3.8em}{3.7em}}

\renewcommand\tableofcontents{%
    \section*{\large\contentsname
        \@mkboth{%
          \MakeUppercase\contentsname}{\MakeUppercase\contentsname}}%
       {\baselineskip=15pt plus 2pt minus 1pt
    \@starttoc{toc}}%
}

\renewenvironment{thebibliography}[1]
     {\baselineskip=16pt plus 2pt minus 1pt
      \section*{\large\refname
        \@mkboth{\MakeUppercase\refname}{\MakeUppercase\refname}}%
     \list{\@biblabel{\@arabic\c@enumiv}}%
           {\settowidth\labelwidth{\@biblabel{#1}}%
            \leftmargin\labelwidth
            \advance\leftmargin\labelsep
            \@openbib@code
            \usecounter{enumiv}%
            \let\p@enumiv\@empty
            \renewcommand\theenumiv{\@arabic\c@enumiv}}%
      \sloppy
      \clubpenalty4000
      \@clubpenalty \clubpenalty
      \widowpenalty4000%
      \sfcode`\.\@m
 \catcode`\^^M=10%
}

\newcommand{\appendices}{
\section*{Appendix}\label{appendices}\setcounter{subsection}{0}
\addcontentsline{toc}{section}{Appendix}
\setcounter{equation}{0}
\makeatletter
\renewcommand{\theequation}{\Alph{subsection}.\arabic{equation}}
\renewcommand{\thesubsection}{\Alph{subsection}}
\@addtoreset{equation}{subsection}
\makeatother
}

\renewcommand{\geq}{\geqslant}

\DeclareMathOperator{\tr}{tr}

\def\XXint#1#2#3{{\setbox0=\hbox{$#1{#2#3}{\int}$}
    \vcenter{\hbox{$#2#3$}}\kern-.5\wd0}}

\newcommand{\alg}[1]{\mathfrak{#1}}

\begin{document}

\begin{titlepage}

\setcounter{page}{0}
\renewcommand{\thefootnote}{\fnsymbol{footnote}}

\begin{flushright}
DMUS--MP--16/04 
\end{flushright}

\vspace{1cm}

\begin{center}

\textbf{\Large\mathversion{bold} Lectures on Classical Integrability}

\vspace{1cm}

{\large Alessandro Torrielli \footnote{{\it E-mail address:\/}
\href{mailto:a.torrielli@surrey.ac.uk}{\ttfamily a.torrielli@surrey.ac.uk}}
} 

\vspace{1cm}

\it Department of Mathematics, University of Surrey\\
Guildford GU2 7XH, United Kingdom

\vspace{1cm}

{\bf Abstract}
\end{center}
\vspace{-.3cm}
\begin{quote}
We review some essential aspects of classically integrable systems. The detailed outline of the lectures consists of
\begin{enumerate}
\item Introduction and motivation, with historical remarks;
\item Liouville theorem and action-angle variables, with examples (harmonic oscillator, Kepler problem);
\item Algebraic tools: Lax pairs, monodromy and transfer matrices, classical $r$-matrices and exchange relations, non-ultralocal Poisson brackets, with examples (non-linear Schr\"odinger model, principal chiral field);
\item Features of classical $r$-matrices: Belavin-Drinfeld theorems, analyticity properties, and lift of the classical structures to quantum groups;
\item Classical inverse scattering method to solve integrable differential equations: soliton solutions, spectral properties and the Gel'fand-Levitan-Marchenko equation, with examples (KdV equation, Sine-Gordon model).
\end{enumerate}    

Prepared for the Durham {\it Young Researchers Integrability School}, organised by the GATIS network. This is part of a collection of lecture notes.

\bigskip

\vfill
\noindent Durham, UK - July 2015

\end{quote}

\setcounter{footnote}{0}\renewcommand{\thefootnote}{\arabic{thefootnote}}

\end{titlepage}

\tableofcontents

\bigskip
\bigskip
\hrule
\bigskip
\bigskip

\section{Introduction and Motivation}

In this section, we give a very short introduction and motivation to the subject. It would be a titanic effort to provide even just an adequate list of references. Therefore we will simply mention a few, relying for the others on the many reviews and books available by now on this topic.

\smallskip 

Let us point out that the main source for these lecture notes is the classic textbook by Babelon, Bernard and Talon \cite{BabelonBernardTalon}. The reader is also encouraged to consult \cite{FaddeTak,Novikov:1984id,Gleb,ETH}.

\subsection{Historical Remarks}

Soon after the formulation of Newton's equations, people have tried to find exact solutions for interesting non-trivial cases. The Kepler problem was exactly solved by Newton himself. Nevertheless, apart from that, only a handful of other problems could be treated exactly.

In the 1800s Liouville refined the notion of {\it integrability}
for Hamiltonian systems, providing a general framework for solving particular dynamical systems {\it by quadratures}. However, it was not until the 1900s that a more or less systematic method was developed. This goes under the name of the {\it classical inverse scattering method}. It was invented by Gardner, Green, Kruskal and Miura in 1967, when they successfully applied it to solve the Korteweg-deVries (KdV) equation of fluid mechanics \cite{GardnerGreeneKruskalMiura}, and it was further developed in \cite{ZS1,ZS2}.

The quantum mechanical version of the inverse scattering method was then elaborated during the following decade by the Leningrad -- St.~Petersburg school (see for instance \cite{Evgeny}), with Ludwig Faddeev as a head, and among many others Korepin, Kulish, Reshetikhin, Sklyanin, Semenov Tian-Shansky and Takhtajan. They established a systematic approach to integrable\footnote{The term {\it integrable} referring to a field-theoretical (infinite-dimensional) model originates from \cite{FaddeevZakharov}.} quantum mechanical systems which makes connection to Drinfeld and Jimbo's theory of quantum groups \cite{Drinfeld}, paving the way to the algebraic reformulation of the problem.
This approach has the power of unifying in a single mathematical framework integrable quantum field theories (cf. \cite{ZamolodchikovZamolodchikov}) together with lattice spin systems (cf. \cite{Baxter}). 

It is probably appropriate to mention that {\it integrability} is still not the same as {\it solvability}. The fact that the two things very often go together is certainly what makes integrable theories so appealing, nevertheless one requires a distinction. There exist integrable systems which one cannot really solve to the very end, as well as exactly solvable systems which are not integrable\footnote{The first statement is particularly relevant for quantum integrable systems on compact domains, such as those described in Stijn van Tongeren's lectures at this school \cite{Stijn}. For such systems, the spectrum is encoded in a set of typically quite complicated integral equations - {\it Thermodynamic Bethe Ansatz} - which can often only be solved via numerical methods. The second statement can instead be exemplified by those {\it dissipative} mechanical systems - such as a free-falling particle subject to air resistance - which are not integrable (and, in fact, not {\it conservative} either) but admit an exact solution. A more elaborated example is however the one of {\it exactly solvable chaotic systems} - see for instance \cite{Ulam}.}. {\it Solvability} ultimately depends on one's ability and computational power. {\it Integrability} rather refers to the property of a system to exhibit regular (quasi-periodic) {\it vs.} chaotic behaviour, and to its conservation laws. This enormously facilitates, and, most importantly, provides {\it general mathematical methods} for the exact solution of the problem. The same holds for the quantum mechanical version, where integrability implies very specific properties of the scattering theory and of the spectrum, as it will be amply illustrated during this school. 

Nowadays, integrability appears in many different contexts within mathematics and mathematical physics. The list below includes only a very small subset of all the relevant research areas.

\begin{enumerate}

\item Classical integrability\footnote{We would like to point out a conjecture put forward by Ward in 1985 \cite{Ward:1985gz} (see also \cite{Ablowitz:2003bv}): {\it ``... many (and perhaps all?) of the ordinary and partial differential equations that are
regarded as being integrable or solvable may be obtained from the self-dual gauge [Yang-Mills, ndr.] field
equations (or its generalisations) by [$4d$ dimensional, ndr.] reduction."} The vast freedom in the Lax pair formulation of the self-dual Yang-Mills equations is due to the arbitrariness of the gauge group. We thank M. Wolf for the information.}

Theory of PDEs, Differential Geometry, General Relativity, Fluid Mechanics.

\item Quantum integrability:

Algebra, Knot Theory, Condensed Matter Physics, String Theory.

\end{enumerate}

It would be impossible to do justice to all the applications of integrability in mathematics and physics; that is why we will just end this introduction with an {\it Ipse dixit}.

\subsection{Why integrability?}

L. Faddeev once wrote \cite{Faddeev} 

{\small
\begin{enumerate}

\item[] 

{\it ``One can ask, what is good in $1 + 1$ models, when our space–time is $3 + 1$ dimensional. There are several particular answers to this question.

\begin{enumerate}
\item The toy models in $1 + 1$ dimension can teach us about the realistic field-theoretical models in a nonperturbative way. Indeed such phenomena as renormalisation, asymptotic freedom, dimensional transmutation (i.e. the appearance of mass via the regularisation parameters) hold in integrable models and can be described exactly.
\item There are numerous physical applications of the $1 + 1$ dimensional models in the condensed matter physics.
\item The formalism of integrable models showed several times to be useful in the modern string theory, in which the world-sheet is $2$ dimensional anyhow. In particular the conformal field theory models are special massless limits of integrable models.
\item The theory of integrable models teaches us about new phenomena, which were not appreciated in the previous developments of Quantum Field Theory, especially in connection with the mass spectrum.
\item I cannot help mentioning that working with the integrable models is a delightful pastime. They proved also to be very successful tool for the educational purposes.
\end{enumerate}
These reasons were sufficient for me to continue to work in this domain for the last 25 years (including 10 years of classical solitonic systems) and teach quite a few followers, often referred to as Leningrad - St. Petersburg school."
}
\end{enumerate}

\section{Integrability in Classical Hamiltonian Systems}

In this section we review the notion of integrability for a Hamiltonian dynamical system, and how this can be used to solve the equations of motion. 

\subsection{Liouville Theorem}

Let us take a Hamiltonian dynamical system with a $2d$-dimensional phase space $M$ parameterised by the canonical variables
\begin{eqnarray}
\label{var}
(q_\mu, p_\mu), \qquad\mu=1,...,d.
\end{eqnarray} 
Let the Hamiltonian function be $H(q_\mu,p_\mu)$, where $q_\mu,p_\mu$ denotes the collection of variables (\ref{var}). We also require the Poisson brackets to be 
\begin{eqnarray}
\{ q_\mu, q_\nu\} = \{p_\mu, p_\nu\} = 0, \qquad \{ q_\mu, p_\nu\} = \delta_{\mu \nu}, \qquad \forall \, \, \mu,\nu = 1,...,d. 
\end{eqnarray}
One calls the system {\it Liouville integrable} if one can find $d$ independent conserved quantities $F_\mu$, $\mu=1,...,d$, in involution, namely
\begin{eqnarray}
\label{invo}
\{F_\mu, F_\nu\} = 0, \qquad \forall \, \, \, \mu,\nu = 1,...,d.
\end{eqnarray} 
Independence here refers to the linear independence of the set of one-forms $dF_\mu$. Note that, since $d$ is the maximal number of such quantities, and since conservation of all the $F_\mu$ means $\{H,F_\mu\}=0 \, \, \forall \, \mu=1,...,d$, then one concludes that 
\begin{eqnarray}
\label{funof}
H = H(F_\mu),
\end{eqnarray} 
{\it i.e.} the Hamiltonian itself is a function of the quantities $F_\mu$.

\newtheorem*{Liouville Theorem}{Theorem (Liouville)}

\begin{Liouville Theorem}
The equations of motion of a Liouville-integrable system can be solved {\it ``by quadratures"} \footnote{We may consider this as a synonym of {\it ``by straightforward integration"}.}.
\end{Liouville Theorem}

\begin{proof}
Let us take the canonical one-form
\begin{eqnarray}
\alpha \equiv \sum_{\mu=1}^d p_\mu \, dq_\mu
\end{eqnarray}
and consider the $d$-dimensional level submanifold 
\begin{eqnarray}
M_f \equiv \{(q_\mu,p_\mu) \in M | F_\mu =f_\mu\}
\end{eqnarray}
for some constants $f_\mu$, $\mu=1,...,d$.
Construct the function
\begin{eqnarray}
\label{S}
S \equiv \int_C \alpha \, = \, \int_{q_0}^q \sum_{\mu=1}^d p_\mu \, dq_\mu,
\end{eqnarray}
where the open (smooth, non self-intersecting) path $C$ lies entirely in $M_f$. In (\ref{S}), one thinks of the momenta $p_\mu$ as functions of $F_\mu$, which are constant on $M_f$, and of the coordinates (see comments in section \ref{aa}). 

The function $S$ is well-defined as a function of the upper extremum of integration $q$ (with $q_0$ thought as a convenient reference point), because the integral in its definition (\ref{S}) does not depend on the path. This is a consequence (via Stokes' theorem) of the fact that $d \alpha = 0$ on $M_f$. We will now prove that $d \alpha = 0$ on $M_f$. 

\begin{proof}

One has that
\begin{eqnarray}
\omega \equiv d \alpha = \sum_{\mu=1}^d dp_\mu \wedge dq_\mu 
\end{eqnarray}  
is the symplectic form on $M$. Let us denote by $X_\mu$ the Hamiltonian vector field associated to $F_\mu$, acting like
\begin{eqnarray}
X_\mu (g) \equiv \{F_\mu, g\} = \sum_{\nu=1}^d \Bigg(\frac{\partial F_\mu}{\partial q_\nu}\frac{\partial g}{\partial p_\nu} - \frac{\partial F_\mu}{\partial p_\nu}\frac{\partial g}{\partial q_\nu}\Bigg)
\end{eqnarray}
on any function $g$ on $M$. Equivalently, 
\begin{eqnarray}
dF_\mu(.) = \omega (X_\mu,\cdot)
\end{eqnarray}
on vectors of the tangent space to $M$. Then, one has 
\begin{eqnarray}
\label{this}
X_\mu (F_\nu) = \{F_\mu, F_\nu\} = 0 \qquad \forall \, \, \mu,\nu=1,...,d
\end{eqnarray}
because of (\ref{invo}). 

Eq. (\ref{this}) then implies that the $X_\mu$ are tangent to the level manifold $M_f$ and form a basis for the tangent space to $M_f$, since the $F_\mu$ are all independent. 
One therefore obtains that
\begin{eqnarray}
\omega (X_\mu, X_\nu) = dF_\mu (X_\nu) = X_\nu (F_\mu) = 0,
\end{eqnarray} 
hence $\omega = 0$ on $M_f$. 
\end{proof}

At this point, we simply regard $S$ as a function of $F_\mu$ and of the upper integration point $q_\mu$, and conclude that 
\begin{eqnarray}
d S \, =  \, \sum_{\nu=1}^d \, p_\nu \, d q_\nu \, + \, \sum_{\nu=1}^d \frac{\partial S}{\partial F_\nu} \, d F_\nu \equiv   \, \sum_{\nu=1}^d \, p_\nu \, d q_\nu \, + \, \sum_{\nu=1}^d \psi_\nu \, d F_\nu,
\end{eqnarray}
where we have defined
\begin{eqnarray}
\psi_\mu \equiv \frac{\partial S}{\partial F_\mu}.
\end{eqnarray}
From $d^2 S = 0$ we deduce
\begin{eqnarray}
\omega = \sum_{\mu=1}^d d F_\mu \, \wedge \, d \psi_\mu
\end{eqnarray}
which shows that the transformation $(q_\mu,p_\mu) \rightarrow (\psi_\mu,F_\mu)$ is canonical. Moreover, all the new momenta $F_\mu$ are constants of motion, hence the time-evolution of the new coordinates is simply
\begin{eqnarray}
\label{conse}
\frac{d \psi_\mu}{dt} \, = \, \frac{\partial H}{\partial F_\mu} \, = \, \mbox{constant in time}
\end{eqnarray}
due to (\ref{funof}) and conservation of the $F_\mu$. The evolution of the new coordinates is therefore linear in time, as can be obtained by straightforward integration of (\ref{conse}) (namely, performing one {\it quadrature}).
\end{proof}

\subsection{Action-angle variables}
\label{aa}

The manifold $M_f$ defined by the equations $F_\mu(q_\nu, p_\nu ) = f_\mu$ typically displays non-trivial cycles, corresponding to a non-trivial topology, therefore the new coordinates $\psi_\mu$ are in principle multi-valued. For instance, the $d$-dimensional anisotropic harmonic oscillator admits $d$ conserved quantites in involution:
\begin{eqnarray}
\label{harmo}
F_\mu = \frac{1}{2} (p_\mu^2 + \omega_\mu^2 \, q_\mu^2)
\end{eqnarray}
(where we have set the mass equal to 1 for convenience). One can see from (\ref{harmo}) that the expression for $p_\mu = p_\mu (q_\nu, F_\rho)$, which is needed to construct $S$ in (\ref{S}), allows for two independent choices of sign. The level manifold $M_f$ is diffeomorphic to a $d$-dimensional torus\footnote{Notice that this is not a peculiar situation, rather it is quite generic. Under the assumption of compactness and connectedness, it is actually guaranteed by the {\it Arnold-Liouville} theorem.}.
 
We shall put ourselves in the situation where the manifold $M_f$ has exactly $d$ non-trivial cycles $C_\mu$, $\mu=1,...,d$. The {\it action} variables are then defined by
\begin{eqnarray}
\label{I}
I_\mu \, = \frac{1}{2 \pi} \, \oint_{C_\mu} \, \alpha,
\end{eqnarray} 
depending only on the $F_\nu$. We can therefore regard $S$ in (\ref{S}) as 
\begin{eqnarray}
S = S(I_\mu, q_\nu).
\end{eqnarray}
If we now define 
\begin{eqnarray}
\theta_\mu \, \equiv \, \frac{\partial S}{\partial I_\mu}
\end{eqnarray}
we have
\begin{eqnarray}
\oint_{C_\mu} d \theta_\nu \, = \, \frac{\partial}{\partial I_\nu} \, \oint_{C_\mu} d S \, = \, \frac{\partial}{\partial I_\nu} \, \oint_{C_\mu} \sum_{\rho=1}^d \frac{\partial S}{\partial q_\rho} \, d q_\rho= \, \frac{\partial}{\partial I_\nu} \, \oint_{C_\mu} \alpha, 
\end{eqnarray}
where we have used that $dI_\mu=0$ along the contour since the $I_\mu$ are constant on $M_f$, and that $\frac{\partial S}{\partial q_\mu} = p_\mu$. From the very definition (\ref{I}) we then get
\begin{eqnarray}
\oint_{C_\mu} d \theta_\nu \, = 2 \pi \, \delta_{\mu \nu},
\end{eqnarray}
displaying how every variable $\theta_\mu$ changes of an amount of $2 \pi$ along their corresponding cycle $C_\mu$. This shows that the $\theta_\mu$ are {\it angle} variables parameterising the $d$-dimensional torus.

\bigskip

{\it \underline{Examples}}

\medskip

\begin{itemize}

\item Let us consider a one-dimensional harmonic oscillator with $\omega=1$, such that the Hamiltonian reads
\begin{eqnarray}
H = \frac{1}{2}(p^2 + q^2).
\end{eqnarray}
There is one conserved quantity $H=F>0$, with level manifold 
\begin{eqnarray}
M_f = \{(q,p) \, | \, q^2 + p^2 = 2 F \equiv R^2 = \mbox{constant}\} = \{q = R \cos \alpha, \, p = R \sin \alpha, R>0, \alpha \in [0,2 \pi)\}. \nonumber
\end{eqnarray}
We can immediately see that $M_f$ is a circle of radius $R = \sqrt{2 F}$ in phase space. 
If we calculate $S$ from (\ref{S}) we obtain, for $q_0 = R$,
\begin{eqnarray}
S = - R^2 \int_0^\alpha d\alpha' \, \sin^2 \alpha' \, = \, F (- \alpha + \sin\alpha \, \cos \alpha).
\end{eqnarray}
At this point, we need to re-express $S$ as a function of $F$ and $q$, therefore we need to distinguish between the two branches $p>0$ and $p<0$, which introduces a multi-valuedness into $S$:
\begin{eqnarray}
\alpha = \arccos \frac{q}{R}, \qquad \mbox{if} \, \, \, \alpha \in [0,\pi], \qquad \alpha = 2 \pi - \arccos \frac{q}{R}, \qquad \mbox{if} \, \, \, \alpha \in (\pi, 2\pi).
\end{eqnarray}
If we define a function
\begin{eqnarray}
g(q,F) \equiv \frac{q}{2} \sqrt{2 F - q^2} - F \arccos \frac{q}{\sqrt{2 F}}, \qquad q \in \Big[-\sqrt{2 F},\sqrt{2 F}\Big],
\end{eqnarray}
then
\begin{eqnarray}
S = g(q,F), \qquad \mbox{if} \, \, \, p\geq 0, \qquad S = - 2 \pi F - g(q,F), \qquad \mbox{if} \, \, \, p<0.
\end{eqnarray}
Let us now calculate the new coordinate $\psi = \frac{\partial S}{\partial F}$. We obtain
\begin{eqnarray}
\psi = - \arccos \frac{q}{\sqrt{2 F}}, \qquad \mbox{if} \, \, \, p\geq 0, \qquad \psi = - 2 \pi + \arccos \frac{q}{\sqrt{2 F}}, \qquad \mbox{if} \, \, \, p<0,
\end{eqnarray} 
from which we conclude
\begin{eqnarray}
\psi = - \alpha.
\end{eqnarray}
We can immediately verify that $\psi$ depends linearly on time, because we know that the solution to the equations of motion is
\begin{eqnarray}
q = \sqrt{2 F} \cos (t + \phi)
\end{eqnarray}
for an initial phase $\phi$. We also find from (\ref{I})
\begin{eqnarray}
I = \frac{1}{2 \pi} S_{\alpha = 2 \pi} \, = \, - F,
\end{eqnarray}
hence the angle variable is given by
\begin{eqnarray}
\theta = \frac{\partial S}{\partial I} = - \frac{\partial S}{\partial F} = - \psi = \alpha.
\end{eqnarray}
The phase space is foliated by circles of constant energy, and the action-angle variables are the polar coordinates (see Figure \ref{fig:foliation}).

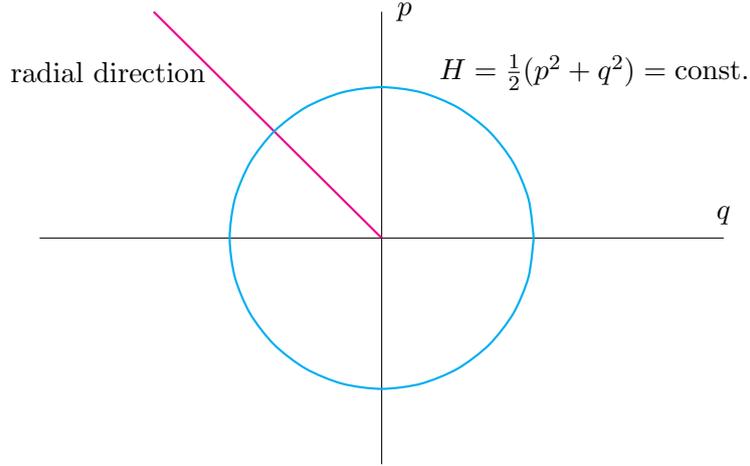
\begin{figure}[t]
  \centering
  \begin{tikzpicture}
\draw[-] (-4.5cm,0) to (4.5cm,0);
\draw[-] (0,-3cm) to (0,3cm);
\draw[scale=1,domain=0:3,smooth,variable=\x, color=magenta,thick,-] plot ({-\x},{\x});
\draw[scale=1,domain=0:6.283,smooth,variable=\x, color=cyan,thick] plot ({2*cos(\x r)},{2*sin(\x r)});
\node[] at (-3.6,2.2) {radial direction};
\node[] at (4.5,0.3) {$q$};
\node[] at (0.3,3) {$p$};
\node[] at (2.8,2.2) {$H=\frac{1}{2}(p^2+q^2)=\mathrm{const.}$};
\end{tikzpicture}
  \caption{Foliation of phase space in circles of constant energy.}
  \label{fig:foliation}
\end{figure}

\smallskip

\item  {\it The Kepler Problem}

The Kepler problem can generally be formulated as the motion of a three-dimensional particle of mass $m$ in a central potential $V(r)= \frac{\beta}{r}$. 

The dimension of the phase space is 6, therefore finding 3 conserved charges in involution will prove the integrability of the problem. The angular momentum is conserved due to rotational symmetry, hence $J^2$, $J_z$ and $H$ form a Poisson-commuting set. They are also independent, hence the system is integrable.

To exhibit the appropriate change of variables, it is convenient to use polar coordinates $(r,\theta, \phi)$, where $\theta$ is the polar angle $\theta \in [0,\pi]$. The conserved quantities are:
\begin{eqnarray}
H = \frac{1}{2} \Bigg( p_r^2 \, + \frac{p_\theta^2}{r^2} \, + \, \frac{p_\phi^2}{r^2 \sin^2 \theta} \Bigg) \, + \, V(r), \qquad J^2 = p_\theta^2 + \frac{p_\phi^2}{\sin^2 \theta}, \qquad J_z = p_\phi.
\end{eqnarray}
These relations can easily be inverted, and then plugged into the expression for the generating function $S$ (\ref{S}), with sign multi-valuedness. Noticing that, in this particular example, each polar momentum only depends on its conjugated polar coordinate, one obtains 
\begin{eqnarray}
S = \int^r dr \sqrt{2\bigg( H - V(r)\bigg) - \frac{J^2}{r^2}} + \int^\theta d\theta \sqrt{J^2 - \frac{J_z^2}{\sin^2 \theta}} + \int^\phi d\phi \, J_z.
\end{eqnarray}
The new coordinate variables are
\begin{eqnarray}
\label{firsttwo}
\psi_{J_z} = \frac{\partial S}{\partial J_z}, \qquad \psi_{J^2} = \frac{\partial S}{\partial J^2}, \qquad \psi_{H} = \frac{\partial S}{\partial H}.
\end{eqnarray}
As a consequence of (\ref{conse}), the first two coordinates in (\ref{firsttwo}) are simply constant, while the third one obeys
\begin{eqnarray}
t - t_0 = \psi_H = \int^r dr \frac{1}{\sqrt{2\bigg( H - V(r)\bigg) - \frac{J^2}{r^2}}},
\end{eqnarray}
which is the standard formula for Keplerian motion. In order to complete the analysis, we should now define the action-angle variables and determine the frequencies of angular motion along the torus. Let us remark that this necessitates the explicit use of the torus cycles, as dictated by (\ref{I}).  

Notice that more independent quantities are conserved besides $J^2$, $J_z$ and $H$, namely also $J_x$, $J_y$ and the {\it Laplace-Runge-Lenz} vector:
\begin{eqnarray}
\label{RL}
\vec{A} = \vec{p} \times \vec{J} + m \, \beta \, \hat{r}, 
\end{eqnarray}
where $\hat{r}$ is the unit vector in the radial direction.
When the total number of independent conserved quantities (the $d$ ones which are in involution plus the remaining ones) equals $d+m$, with $0<m<d-1$, we will call the system {\it super-integrable}. When $m=d-1$, we call the system {\it maximally super-integrable}. In the case of the Kepler problem, we have found 8 conserved quantities, namely $E, J^2, \vec{J}$ and $\vec{A}$. However, only five are independent, as we have three relations:
\begin{eqnarray}
J^2 = \sum_{\mu=1}^3 J_\mu^2, \qquad \vec{A} \cdot \vec{J} = 0, \qquad A^2 = m^2 \beta^2 \, + \, 2 m E \, J^2.  
\end{eqnarray}
The Kepler problem is therefore maximally super-integrable, since $d=3$. 

\end{itemize}

Let us conclude this section by saying that, {\it locally} in phase space, one can reproduce much of the construction we have outlined for generic Hamiltonian systems, which might raise doubts about the peculiarity of integrability. The crucial distinction is that, for integrable systems, the procedure we have described extends {\it globally}. In particular, one has a global foliation of phase space by the $M_f$ submanifolds, and a nice global geometric structure arising from this analysis.  

\section{Algebraic Methods}

In this section we introduce the concepts of Lax pair, Monodromy and Transfer matrix, and Classical $r$-matrix. These quantities prominently feature in the so-called Inverse Scattering Method, which begins by recasting the Poisson brackets of the dynamical variables of a classically integrable system in a form which is most suitable for displaying the structure of its symmetries. At the end of this section we will briefly comment on the issue of {\it (non) ultra-locality} of the Lax-pair Poisson brackets.

\subsection{Lax pairs}

Suppose you can find two matrices $L,M$ such that Hamilton's equations of motion can be recasted in the following form:
\begin{eqnarray}
\label{Lax}
\frac{dL}{dt} \, = \, [M,L].
\end{eqnarray} 
The two matrices $L,M$ are said to form a {\it Lax pair}. From (\ref{Lax}), we can immediately obtain a set of conserved quantites:
\begin{eqnarray}
\label{natural}
O_n \, \equiv \, \tr L^n, \qquad \frac{d O_n}{dt} \, = \, \sum_{i=0}^{n-1} \tr \, L^{i} \, [M,L] L^{n-1-i} \, = \, 0, \qquad \forall \, \, n \, \, \mbox{natural number}, 
\end{eqnarray}
by simply opening up the commutator and using the cyclicity of the trace. Of course not all of these conserved charges are independent. Notice that (\ref{natural}) implies that the eigenvalues $\lambda_\alpha$ of $L$ are conserved in time, since $O_n = \sum_\alpha \lambda_\alpha^n$.

Let us point out that the Lax pair is not unique\footnote{For example, adding constant multiples of the identity to $L$ and $M$ preserves (\ref{Lax}). There also exist particular models where one can describe the system using alternative Lax pairs of different ranks (cf. \cite{MironovEtAl}, pages 2-3).}, as there is at least a gauge freedom
\begin{eqnarray}
L \longrightarrow g \, L \, g^{-1}, \qquad M \longrightarrow g \, M \, g^{-1} + \frac{dg}{dt} \, g^{-1},
\end{eqnarray} 
with $g$ an invertible matrix depending on the phase-space variables.

\bigskip

{\it \underline{Example}}

\medskip

\begin{itemize}

\item A Lax pair for the harmonic oscillator (with mass $m=1$ in appropriate units) can be written down as follows:
\begin{eqnarray}
\label{LMharm}
L = \begin{pmatrix}p&\omega q\\\omega q&-p
\end{pmatrix} = p \, \sigma_3 \, + \, \omega q \, \sigma_1,\qquad M = \begin{pmatrix}0&-\frac{\omega}{2}\\\frac{\omega}{2} &0
\end{pmatrix} = - i \frac{\omega}{2} \, \sigma_2.
\end{eqnarray}
One can immediately check using (\ref{LMharm}) that the only non-zero components of (\ref{Lax}) - those along $\sigma_1$ and $\sigma_3$ - are equivalent to Hamilton's equations of motion. 
\end{itemize}

We will regard $L$ and $M$ as elements of some matrix algebra $\mathfrak{g}$, with the matrix entries being functions on phase space. For example, in the case of the harmonic oscillator (\ref{LMharm}) one can see that $\alg{g}$ is the complexification of the $su(2)$ Lie algebra, which is isomorphic to $sl(2,\mathbb{C})$. 

Even if we assume that we have found a Lax pair, and that we can obtain $d$ independent conserved quantities, this does not yet guarantee their {\it involutivity}. Hence, we have not yet secured integrability. For that, we need an extra ingredient. We introduce the following notation:
\begin{eqnarray}
\label{conve}
X_1 \equiv X \otimes 1, \qquad X_2 \equiv 1 \otimes X
\end{eqnarray}
as elements of $\mathfrak{g} \otimes \mathfrak{g}$. Then, one has the following 

\newtheorem*{Constant r Theorem}{Theorem}

\begin{Constant r Theorem}
\label{costante}
The eigenvalues of $L$ are in involution iff there exists an element $r_{12} \in \mathfrak{g} \otimes \mathfrak{g}$, function of the phase-space variables, such that
\begin{eqnarray}
\label{constantr}
\{ L_1,L_2\} \, = \, [r_{12}, L_1] - [r_{21}, L_2],
\end{eqnarray} 
where $r_{21} = \Pi \circ r_{12}$, $\Pi$ being the permutation operator acting on the two copies of $\mathfrak{g} \otimes \mathfrak{g}$. 
\end{Constant r Theorem}
The proof can be found in \cite{BabelonBernardTalon}.

In order for the Jacobi identity to hold for the Poisson bracket (\ref{constantr}) one needs to impose the following condition, defined on the triple tensor product $\alg{g} \otimes \alg{g} \otimes \alg{g}$: 
\begin{eqnarray}
\label{then}
&&[L_1, [r_{12},r_{13}] + [r_{12},r_{23}] + [r_{32},r_{13}] + \{ L_{2}, r_{13} \} - \{ L_{3}, r_{12} \}] \, + \nonumber\\
&&[L_2, [r_{13},r_{21}] + [r_{23},r_{21}] + [r_{23},r_{31}] + \{ L_{3}, r_{21} \} - \{ L_{1}, r_{23} \}] \, + \nonumber\\
&&[L_3, [r_{31},r_{12}] + [r_{21},r_{32}] + [r_{31},r_{32}] + \{ L_{1}, r_{32} \} - \{ L_{2}, r_{31} \}] \, = \, 0. 
\end{eqnarray}
One can see from here that, if $r_{12}$ is a constant independent of the dynamical variables, and if we furthermore require that
\begin{eqnarray}
r_{12} = - r_{21},
\end{eqnarray}
then we see that a {\it \underline{sufficient}} condition for the Jacobi identity to be satisfied is 
\begin{eqnarray}
\label{cYBE}
[r_{12},r_{13}] + [r_{12},r_{23}] + [r_{13},r_{23}] \, = \, 0. 
\end{eqnarray}
We call such an $r$ a {\it constant classical $r$-matrix}, and (\ref{cYBE}) the {\it classical Yang-Baxter equation (CYBE)}.

Notice that another sufficient condition would be to have (\ref{cYBE}), but with a Casimir element instead of zero on the right hand side. This {\it modified Yang-Baxter equation} would lead us to a more general mathematical setting, which however goes beyond the scope of these lectures. 

\bigskip

{\it \underline{Examples}}

\medskip

\begin{itemize}

\item The following matrix is a constant solution of the CYBE:
\begin{eqnarray}
r = e \otimes h - h \otimes e, \qquad [h,e]=e.
\end{eqnarray}
The algebra it is based upon is the triangular Borel subalgebra of $sl(2)$ generated by the Cartan element $h$ and one of the roots, here denoted by $e$. 

\bigskip

\item The matrix $r_{12}$ for the harmonic oscillator is non-constant (sometimes such $r$-matrices are called {\it dynamical} as opposed to the constant ones), and it reads
\begin{eqnarray}
\label{rarmo}
r_{12} = - \frac{\omega}{4 H} \, \begin{pmatrix}0&1\\-1 &0
\end{pmatrix} \otimes L = - \frac{i \omega}{4 H} \,  \sigma_2 \otimes L,
\end{eqnarray}
with $H=\frac{1}{2}(p^2 + \omega^2 q^2)$ being the energy, and $\{q,p\} =1$ the canonical Poisson bracket. The eigenvalues of $L$ in (\ref{LMharm}) are $\pm 2 H$.

\end{itemize}

\bigskip

The most interesting case for our purposes will be when the Lax pair depends on an additional complex variable $u$, called the {\it spectral parameter}. This means that in some cases we can find a {\it family} of Lax pairs, parameterised by $u$, such that the equations of motion are equivalent to the condition (\ref{Lax}) 
{\it for all values of $u$}. We will see that this fact has significant consequences for the inverse scattering method. Therefore, we are going to put ourselves in this situation from now on.

\bigskip

{\it \underline{Example}}

\medskip

\begin{itemize}

\item A Lax pair for the Kepler problem reads \cite{LaxKepler}
\begin{eqnarray}
\label{LaxKep2}
L = \frac{1}{2} \begin{pmatrix}-\chi_u[\vec{r}, \frac{d}{dt}\vec{r}\, ]&&\chi_u[\vec{r}, \vec{r}\, ]\\-\chi_u[\frac{d}{dt}\vec{r}, \frac{d}{dt}\vec{r}\, ] &&\chi_u[\vec{r},\frac{d}{dt}\vec{r}\, ]\end{pmatrix}, \qquad M = \begin{pmatrix}0&&1\\M_0&&0\end{pmatrix},
\end{eqnarray}
where one has defined
\begin{eqnarray}
\chi_u[\vec{a},\vec{c} \,] \equiv \sum_{\mu=1}^3  \frac{a_\mu c_\mu}{u - u_\mu}, \qquad M_0 \, \vec{r} \equiv - \frac{1}{m} \, \nabla V(r) 
\end{eqnarray}
and $V(r) = \frac{\beta}{r}$. This Lax pair depends on three complex variables $u_\mu$, $\mu=1,2,3$, besides $u$ which we take as a spectral parameter. Laurent-expanding Eq. (\ref{Lax}) in $u$, one recovers the full set of Newton's equations $m \frac{d^2}{dt^2} \vec{r} = - \nabla V(r) = \beta \frac{\vec{r}}{r^3}$.

\end{itemize}

\subsection{Field theories. Monodromy and Transfer matrices.} 

An important step we need to take is to generalise what we have reviewed so far for classical finite-dimensional dynamical systems, to the case of classical {\it field theories}. We will restrict ourselves to two-dimensional field theories, meaning one spatial dimension $x$ and one time $t$. This means that we will now have equations of motion obtained from a classical field theory Lagrangian (Euler-Lagrange equations).

The notion of integrability we gave earlier, based on the Liouville theorem, is inadequate when the number of degrees of freedom becomes infinite, as it is the case for field theories. What we will do is to adopt the idea of a Lax pair, suitably modifying its definition, as a starting point to define an integrable field theory.

Suppose you can find two (spectral-parameter dependent) matrices $L,M$ such that the Euler-Lagrange equations of motion can be recasted in the following form:
\begin{eqnarray}
\label{Lax2}
\frac{\partial L}{\partial t} - \frac{\partial M}{\partial x} \, = \, [M,L].
\end{eqnarray} 
We will call such field theories {\it classically integrable.}

The condition (\ref{Lax2}) is also the compatibility condition for the following {\it auxiliary linear problem}:
\begin{eqnarray}
\label{aux}
(\partial_x - L) \Psi = 0, \qquad  (\partial_t - M) \Psi = 0,
\end{eqnarray}
as can be seen by applying $\partial_t$ to the first equation, $\partial_x$ to the second equation, subtracting the two and using (\ref{aux}) one more time. We will make use of the auxiliary linear problem later on, when we will discuss solitons.

The two matrices $L,M$ in (\ref{Lax2}) are also said to form a {\it Lax pair}, and one can in principle obtain a sequence of conserved quantities for the field theory by following a well-defined procedure. Such a procedure works as follows.

Let us introduce the so-called {\it monodromy matrix} 
\begin{eqnarray}
T(u) = P \exp \Bigg[\int_{s_-}^{s_+} L(x,t,u) dx \Bigg]
\end{eqnarray}
where $P$ denotes a path-ordering with greater $x$ to the left, ${s_-}$ and ${s_+}$ are two points on the spatial line, and $u$ is the spectral parameter. This object can be thought of as implelementing a parallel transport along the segment $[s_-, s_+]$, in accordance with the fact that the Lax pair can be thought of as a connection. 

If we differentiate $T(u)$ with respect to time, we get
\begin{eqnarray}
\partial_t \, T \, &=& \, \int_{s_-}^{s_+} dx \, P \exp \Bigg[\int_x^{s_+} L(x',t,u) dx' \Bigg] \, \, \big[\partial_t \, L(x,t,u)\big] \, \, P \exp \Bigg[\int_{s_-}^x L(x',t,u) dx' \Bigg]\nonumber\\
&=& \, \int_{s_-}^{s_+} dx \, P \exp \Bigg[\int_x^{s_+} L(x',t,u) dx' \Bigg] \, \, \bigg(\frac{\partial M}{dx} \, - \, [L,M]\bigg) \, \, P \exp \Bigg[\int_{s_-}^x L(x',t,u) dx' \Bigg]\nonumber\\
&=& \, \int_{s_-}^{s_+} dx \, \partial_x \Bigg(P \exp \Bigg[\int_x^{s_+} L(x',t,u) dx' \Bigg] \, \, M \, \, P \exp \Bigg[\int_{s_-}^x L(x',t,u) dx' \Bigg]\Bigg)\nonumber\\
&=& M({s_+}) \, T - T \, M({s_-}),
\end{eqnarray}
where at some stage we have used (\ref{Lax2}). At this point, if we consider pushing ${s_-}$ and ${s_+}$ towards the extrema $S_\pm$ of the spatial domain, respectively, and assume for definiteness this domain to be the compact segment $[0,2 \pi]$ with periodic boundary conditions on the fields, we obtain
\begin{eqnarray}
\partial_t \, T \, = \, [M (0,t,u), T].
\end{eqnarray}
This implies that the trace of $T$, called the {\it transfer matrix}
\begin{eqnarray}
\label{Taylor}
\mathfrak{t}(u) \, \equiv \tr T(u),
\end{eqnarray}
is conserved {\it for all $u$}. By expanding in $u$, one obtains a family of conserved charges, which are the coefficients of the expansion. For instance, if $\mathfrak{t}(u)$ is analytic near the origin, one Taylor-expands
\begin{eqnarray}
\mathfrak{t}(u) = \sum_{n\geq 0} Q_n \, u^n, \qquad \partial_t Q_n = 0, \, \, \forall \, n \geq 0.
\end{eqnarray}
This forms the starting point for the construction of the integrable structure.
 
\bigskip

{\it \underline{Examples}}

\medskip

\begin{itemize}

\item The Non-linear Schr\"odinger (NLS) model\footnote{We will follow \cite{Evgeny} in this example.} is the classical non-relativistic $1+1$ dimensional field theory with Hamiltonian
\begin{eqnarray}
\label{NLS}
H = \int_{-\infty}^\infty dx \, \bigg( \left\vert \frac{\partial \psi}{\partial x} \right\vert^2 + \kappa  |\psi |^4 \bigg)
\end{eqnarray}
for a complex field $\psi(x)$ with Poisson brackets
\begin{eqnarray}
\{ \psi(x), \psi^*(y)\} = \delta(x-y),
\end{eqnarray} 
and a real coupling constant $\kappa$.
The name {\it non-linear Schr\"odinger} model comes from the fact that the equations of motion look like
\begin{eqnarray}
i \frac{\partial \psi}{\partial t} = \{H,\psi\} = \frac{\partial^2 \psi}{\partial x^2} + 2 \kappa |\psi |^2 \psi,
\end{eqnarray}
namely they coincide with the Schr\"odinger equation for a non-linear potential depending on the square modulus of the wave function itself.

The Lax pair reads
\begin{eqnarray}
L = \begin{pmatrix}- i \frac{u}{2} &&i \kappa \psi^*\\- i \psi&&i \frac{u}{2} \end{pmatrix}, \qquad M = \begin{pmatrix}i \frac{u^2}{2} + i \kappa |\psi|^2&&\kappa \frac{\partial \psi^*}{\partial x} - i \kappa u \psi^*\\\frac{\partial \psi}{\partial x} + i u \psi&&-i \frac{u^2}{2} - i \kappa |\psi|^2 \end{pmatrix},
\end{eqnarray}
depending on a complex spectral parameter $u$. One can write the monodromy matrix as
\begin{eqnarray}
\label{TNLS}
T(u) = \begin{pmatrix}a(u)&& \kappa \, b^*(u^*)\\b(u)&&a^*(u^*)\end{pmatrix},
\end{eqnarray} 
where $a(u)$ and $b(u)$ admit the following power-series representation:
\begin{eqnarray}
&&a(u) = e^{- i \frac{u}{2} (s_+ - s_-)} \Bigg[ 1 + \sum_{n=1}^\infty \kappa^n \int_{s_+ > \xi_n > \eta_n > \xi_{n-1}...>\eta_1 > s_-} d\xi_1 ... d\xi_n \, d\eta_1 ... d\eta_n \nonumber \\
&&\qquad \qquad \qquad \qquad \qquad e^{i u (\xi_1 + ... + \xi_n - \eta_1 - ... - \eta_n)} \psi^*(\xi_1) ... \psi^*(\xi_n) \, \psi(\eta_1) ... \psi(\eta_n)\Bigg],\nonumber \\
&&b(u) = - i \, e^{i \frac{u}{2} (s_+ + s_-)} \Bigg[ 1 + \sum_{n=1}^\infty \kappa^n \int_{s_+ > \eta_{n+1} > \xi_n > \eta_n > \xi_{n-1}...>\eta_1 > s_-} d\xi_1 ... d\xi_n \, d\eta_1 ... d\eta_{n+1} \nonumber \\
&&\qquad \qquad \qquad \qquad \qquad e^{i u (\xi_1 + ... + \xi_n - \eta_1 - ... - \eta_{n+1})} \psi^*(\xi_1) ... \psi^*(\xi_n) \, \psi(\eta_1) ... \psi(\eta_{n+1})\Bigg].\nonumber
\end{eqnarray}

Some of the conserved charges one obtains by an appropriate expansion\footnote{\label{fooNLS}In this example one actually Laurent-expands \begin{eqnarray}
\label{gene}
\log a(u) = i \kappa \sum_{m=1}^\infty \mathfrak{I}_m u^{-m}.
\end{eqnarray} 
The justification for this is that $a(u)$ itself commutes with the Hamiltonian, and can be proven to be an equally good generating function for the conserved charges \cite{SZ,T}. Moreover, one can show that $\{a(u),a(u')\}=0$, hence the charges generated by  (\ref{gene}) are all in involution with each other. This will be made more systematic in the following sections, where it will be seen to follow from the {\it Sklyanin exchange relations}.}, and in the limit of infinite domain $s_\pm \to \pm \infty$, read
\begin{eqnarray}
\label{char}
&&\mathfrak{I}_1 = \int_{-\infty}^\infty dx |\psi|^2, \qquad \mathfrak{I}_2 = \frac{i}{2} \int_{-\infty}^\infty dx \bigg(\frac{\partial\psi^*}{\partial x} \psi - \psi^* \frac{\partial\psi}{\partial x}\bigg), \qquad \mathfrak{I}_3 = H,\nonumber\\
&&\mathfrak{I}_4 = i \int_{-\infty}^\infty dx \bigg(\frac{\partial^3\psi^*}{\partial x^3} \psi - 3 \frac{\partial\psi^*}{\partial x} \psi |\psi|^2\bigg), \qquad \mbox{etc.}
\end{eqnarray}
Upon quantisation, the first charge corresponds to the particle number, the second one to the momentum, the third one to the Hamiltonian.
 
Action-angle type variables\footnote{Let us quote Sklyanin's original words \cite{Evgeny}: {\it ``The concept of ``action-angle variables" we shall treat here broadly, calling such any canonical variables in which the Hamiltonian $H$ can be written as a quadratic form (and the equations of motion, correspondingly, become linear)."}} for the NLS can be obtained in the following fashion. 
If we define
\begin{eqnarray}
\varphi(u) = \frac{b(u)}{|b(u)|} \sqrt{\frac{\log a(u)}{\kappa \pi}},
\end{eqnarray}
then this new variable is such that
\begin{eqnarray}
\{\varphi(u), \varphi^*(u')\} = i \delta (u - u'), \qquad \frac{d}{dt} \varphi(u) = - i u \, \varphi(u). 
\end{eqnarray}
The second equation guarantees that $\varphi$ is the exponential of an angle variable.  

In the new variables, the infinite tower of conserved charges - cf. (\ref{char}) - collectively read:
\begin{eqnarray}
\mathfrak{I}_m = \int_{-\infty}^\infty d \mu \, \mu^{m-1} |\varphi(\mu)|^2, \qquad m=1,2,...
\end{eqnarray}

 \item The Sine-Gordon equation for a real scalar field $\phi$ in $1+1$ space-time dimensions reads
\begin{eqnarray}
\label{SG}
\partial_t^2 \phi \, - \, \partial_x^2 \phi \, = \, - \, \frac{8 m^2}{\beta} \, \sin (2 \beta \phi),
\end{eqnarray}
for $m$ and $\beta$ two constants. One can recast this equation as a Lax pair condition, with a Lax pair given by
\begin{eqnarray}
&&L = i \, \begin{pmatrix}\frac{\beta}{2} \partial_t \phi &&m u \, e^{i \beta \phi} - \frac{m}{u} e^{- i \beta \phi}\\m u \, e^{- i \beta \phi} - \frac{m}{u} e^{i \beta \phi}&&- \frac{\beta}{2} \partial_t \phi \end{pmatrix},\nonumber\\
&&M = i \, \begin{pmatrix}\frac{\beta}{2} \partial_x \phi &&-m u \, e^{i \beta \phi} - \frac{m}{u} e^{- i \beta \phi}\\-m u \, e^{- i \beta \phi} - \frac{m}{u} e^{i \beta \phi}&&- \frac{\beta}{2} \partial_x \phi \end{pmatrix},
\end{eqnarray}
depending on a spectral parameter $u$.
\end{itemize}

\subsection{Poisson structure and the problem of Non-Ultralocality}

In the spirit of (\ref{constantr}) we now consider the Poisson brackets between two Lax pair elements $L$, this time taken at two distinct positions $x$ and $y$ and for different spectral parameters $u$ and $u'$. Suppose that the canonical Poisson brackets imposed on the fields have the following consequence for $L$:
\begin{eqnarray}
\label{nonconstr}
\{ L_1(x,t,u), L_2(y,t,u')\} \, = \, [r_{12}(u - u'), L_1(x,t,u) + L_2(y,t,u')] \, \delta(x-y),
\end{eqnarray}  
with similar conventions as in (\ref{conve}). Let us also assume that the $r$-matrix $r_{12}(u - u')$ does not itself depend on the fields\footnote{In principle, we shall not dub this $r$-matrix as {\it constant} any longer, because of the dependence on the spectral parameter.}, and satisfies
\begin{eqnarray}
r_{12}(u-u') = - r_{21}(u'-u).
\end{eqnarray}

\newtheorem*{Evgeny}{Theorem (Sklyanin Exchange Relations)}

\begin{Evgeny}
\label{Evgeny1}
Given (\ref{nonconstr}), the Poisson brackets of the monodromy matrix satisfy
\begin{eqnarray}
\label{TEvge}
\{ T_1(u), T_2(u')\} \, = \, [r_{12}(u - u'), T_1(u) T_2(u')].
\end{eqnarray}  
\end{Evgeny}

From this, one can immediately conclude that the conserved charges generated by the transfer matrix are all in involution. Indeed, tracing by $\tr_1 \otimes \tr_2$ both sides of (\ref{TEvge}), one obtains
\begin{eqnarray}
\label{EvgeT}
\{\mathfrak{t}(u), \mathfrak{t}(u')\} = 0,
\end{eqnarray} 
where we have used cyclicity of $\tr_1 \otimes \tr_2$ which is the natural trace operation on $\mathfrak{g} \otimes \mathfrak{g}$. By expanding 
(\ref{EvgeT}) we obtain the desired involution property of the charges. For analytic functions,
\begin{eqnarray}
\mathfrak{t}(u) = \sum_{n\geq 0} Q_n \, u^n, \qquad \{Q_n, Q_m\} = 0 \, \, \, \, \, \forall \, \, m,n \geq 0.
\end{eqnarray} 
The variables $L$ and $T$ can therefore be thought of as the most convenient variables to display the integrable structure of the model. It will not come as a surprise then that quantisation best proceeds from the Sklyanin relations, in what constitutes the backbone of the Quantum Inverse Scattering Method (QISM) \cite{Evgeny}.

With the assumptions described after (\ref{nonconstr}), the Jacobi identity for Sklyanin's exchange relations admits again as a sufficient condition the {\it classical Yang-Baxter equation} with spectral parameter, namely
\begin{eqnarray}
\label{diffYBE}
[r_{12}(u_1 - u_2),r_{13}(u_1 - u_3)] + [r_{12}(u_1 - u_2),r_{23}(u_2 - u_3)] + [r_{13}(u_1 - u_3),r_{23}(u_2 - u_3)] \, = \, 0.\nonumber \\ 
\end{eqnarray}
The Poisson brackets (\ref{nonconstr}) are called {\it ultra-local}, because they only display the Dirac delta function and not its derivatives. Whenever higher derivatives of $\delta(x-y)$ are present, one speaks of {\it non ultra-local} Poisson structures. In the latter case, one cannot obtain a formula like (\ref{TEvge}), and quantisation does not proceed along the standard lines of the QISM. 

\bigskip

\noindent {\it Maillet brackets}

\bigskip 

\noindent There is a situation where a significant amount of progress has been made, despite the presence of non ultra-locality. This is when the Poisson brackets between the spatial component of the Lax pair assume the form 
\begin{eqnarray}
\label{brac}
&&\{ {{L}}_1 (x,t,u) , {{L}}_2 (y,t,u') \} \, = \, \delta(x - y) \, [r_- (u,u'), {{L}}_1 (x,t,u)]\\
&&\qquad + \, \delta(x - y) \, [r_+ (u,u'), {{L}}_2 (y,t,u')] \, + \, \delta'(x - y) \, \Big(r_- (u,u') - r_+(u,u')\Big), \nonumber  
\end{eqnarray}
for a choice of an $(r,s)$-matrix pair satisfying a mixed Yang-Baxter type equation:
\begin{eqnarray}
r_+ \, = \, r+s, \qquad r_- \, = \, r-s,
\end{eqnarray}
\begin{eqnarray}
\label{mcYBE}
&&[(r+s)_{13}(u_1, u_3), (r-s)_{12}(u_1, u_2)] \, + \, [(r+s)_{23}(u_2, u_3), (r+s)_{12}(u_1, u_2)] \, \nonumber\\
&&\qquad \qquad \qquad + \,
[(r+s)_{23}(u_2, u_3), (r+s)_{13}(u_1, u_3)] \, = \, 0
\end{eqnarray}
(this is again to ensure the Jacobi identity of the brackets).
The Poisson brackets for the (classical) monodromy matrix $T\equiv P \exp \int {{L}}$ have been derived from (\ref{brac}) by careful treatment of the ambiguity arising from the non ultra-locality \cite{Maillet1,Maillet2}, and read\footnote{It is important to remark that the brackets (\ref{Mbrac}) do not satisfy the Jacobi identity. This violation is connected to a naive equal-point limiting procedure, occurring when one first evaluates $\{T(u) \otimes 1, 1 \otimes T(u')\}$ for different values of $s_\pm$ in each of the two factors. A careful treatment of this singularity is provided in \cite{Maillet2}, where a way to restore the Jacobi identity is defined via a more elaborated symmetric-limit procedure.}
\begin{eqnarray}
\label{Mbrac}
&&\{T(u) \otimes 1, 1 \otimes T(u')\}\, = \, [r(u,u'), T(u) \otimes T(u')] \\
&&\qquad - \, [1 \otimes T(u')] \, s(u,u') \, [T(u) \otimes 1] \, + \, 
[T(u) \otimes 1] \, s(u,u') [1 \otimes T(u')].\nonumber
\end{eqnarray}
By taking the trace of (\ref{Mbrac}), one can still show that an infinite set of classically conserved charges in involution is generated by $\tr T(u)$. The problem is that no quantisation procedure has been so far fully established for the brackets (\ref{Mbrac}). The quantum S-matrix can nevertheless be fixed by symmetries in most of the interesting cases.

\bigskip

{\it \underline{Example}}

\medskip

\begin{itemize}

\item The Principal Chiral Model provides a standard example of Maillet structure \cite{Maillet1,Niall}. This is the theory of an element $g$ of a compact group G, with Lagrangian
$$
{\cal{L}}=-\frac{1}{2 \gamma} \tr j_\mu \, j^\mu, \qquad j_\mu = (\partial_\mu g) g^{-1}
$$
admitting (left,right) global symmetry $g \rightarrow \big( e^{i f} g, \, g \, e^{i f}\big)$. In Minkowski signature, $x^0 = t$, $x^1 = x$, and conservation of $j$ reads $\partial_\mu j^\mu = \partial_0 j_0 - \partial_1 j_1 = 0$. The constant $\gamma$ is the coupling of the theory.
Both $f$ and the (left,right) currents
$$j_\mu^L = j_\mu = (\partial_\mu g) g^{-1},  \qquad j_\mu^R = \, - g^{-1} \, (\partial_\mu g)$$
belong to the Lie algebra of G. The currents are flat (cf. {\it Maurer-Cartan} one-forms), {\it i.e.}
\begin{eqnarray}
\partial_\mu j^L_\nu - \partial_\nu j^L_\mu - [j^L_\mu, j^L_\nu] = 0, \qquad \partial_\mu j^R_\nu - \partial_\nu j^R_\mu - [j^R_\mu, j^R_\nu] = 0.
\end{eqnarray}
The Lax pair reads
\begin{eqnarray}
L = \frac{u \, j_0 \, + \, j_1}{1 - u^2}, \qquad M = \frac{u \, j_1 \, + \, j_0}{1 - u^2}.
\end{eqnarray}
The $(r,s)$ pair reads
\begin{eqnarray}
&&r (u, u') \, = \frac{1}{2} \, \frac{\zeta(u) + \zeta(u')}{u-u'} \, {\cal{C}_\otimes}\,\qquad s (u, u') \, = \frac{1}{2} \, \frac{\zeta(u) - \zeta(u')}{u-u'} \, {\cal{C}_\otimes},
\end{eqnarray}
with
\begin{eqnarray}
\label{K}
\zeta(u) = \gamma \frac{u^2}{u^2 - 1}, \qquad {\cal{C}_\otimes} = \sum_{a,b} \kappa_{a b} \,\,  t_a \otimes t_b,
\end{eqnarray}
in terms of the Lie algebra generators $t^a$. Indices in (\ref{K}) are saturated with the Killing form $\kappa_{a b}$ (see footnote \ref{Killi} in the following). 

\end{itemize}

\bigskip

We would like to conclude this section by mentioning that the Lax-pair formalism allows to derive a set of relations, called {\it finite-gap equations}, which are the first step to a semi-classical analysis of the spectrum of the integrable system. This is treated in detail in Fedor Levkovich-Maslyuk's lectures at this school \cite{Fedor}, where the finite-gap equations are derived by taking a semi-classical limit of the quantum Bethe ansatz equations.

\section{Classical $r$-matrices}

In this section, we discuss the properties of classical $r$-matrices, most notably their analytic structure and their relation to infinite-dimensional algebras. The highlights of this section will be the famous Belavin-Drinfeld theorems.

\subsection{Belavin-Drinfeld theorems}

Mathematicians have studied classical $r$-matrices and have classified them under specific assumptions. We begin by presenting the most important theorems in this area \cite{BelavinDrinfeld1,BelavinDrinfeld2,BelavinDrinfeld3}.

\newtheorem*{BD1}{Theorem (Belavin Drinfeld I)}

\begin{BD1}
\label{BD11} Let $\alg{g}$ be a finite-dimensional simple\footnote{A Lie algebra is simple when it has no non-trivial ideals, or, equivalently, its only ideals are $\{ 0 \}$ and the algebra itself. An {\it ideal} is a subalgebra such that the commutator of the whole algebra with the ideal is contained in the ideal.} Lie algebra, and $r=r(u_1 - u_2) \in \alg{g}\otimes \alg{g}$ a solution of the (spectral-parameter dependent) classical Yang-Baxter equation (\ref{diffYBE}). Furthermore, assume one of the following three equivalent conditions to hold: 
\begin{itemize}
\item (i) r has at least one pole in the complex plane $u \equiv u_1 - u_2$, and there is no Lie subalgebra $\alg{g}'\subset \alg{g}$ such that $r \in \alg{g}' \otimes \alg{g}'$ for any $u$, 
\item (ii) $r(u)$ has
a simple pole at the origin, with residue proportional to
$\sum_a t_a \otimes t_a$, with $t_a$ a basis in $\alg{g}$ orthonormal with respect to a chosen nondegenerate invariant bilinear form\footnote{Such a residue can be identified with the quadratic Casimir ${\cal{C}_\otimes}$ in $\alg{g} \otimes \alg{g}$.},
\item (iii) the determinant of the matrix $r_{a b}(u)$ obtained from $r(u) = \sum_{a b} r_{a b}(u)\, t_a \otimes t_b$ does not vanish identically. 
\end{itemize}
Under these assumptions, $r_{12}(u) = - r_{21}(- u)$ where $r_{21}(u) = \Pi \, \circ \, r_{12}(u)= \sum_{a b} r_{a b}(u)\, t_b \otimes t_a$, and $r(u)$ can be extended meromorphically to the entire
$u$-plane. All the poles of $r(u)$ are simple, and they form a lattice $\Gamma$.
One has three possible equivalence classes of solutions: ``elliptic" - when $\Gamma$ is a
two-dimensional lattice -, ``trigonometric" - when $\Gamma$ is a one-dimensional array -, or ``rational"- when $\Gamma = \{0\}$-, respectively.
\end{BD1}

The assumption of difference-form is not too restrictive, thanks to the following theorem by the same authors \cite{BelavinDrinfeld3}:

\newtheorem*{BD2}{Theorem (Belavin Drinfeld II)}

\begin{BD2}
\label{BD22} Assume the hypothesis of Belavin-Drinfeld I theorem but $r=r(u_1,u_2)$ {\bf not} to be of difference form, with the classical Yang-Baxter equation being the natural generalisation of (\ref{diffYBE}):
\begin{eqnarray}
\label{ndiffYBE}
[r_{12}(u_1,u_2),r_{13}(u_1,u_3)] + [r_{12}(u_1,u_2),r_{23}(u_2,u_3)] + [r_{13}(u_1,u_3),r_{23}(u_2,u_3)] \, = \, 0.
\end{eqnarray}
Now the three statements (i), (ii) and (iii) are not any longer immediately equivalent, and we will only retain (ii).
Assume the dual Coxeter number\footnote{\label{Killi}The dual Coxeter number $\alg{c}_2$ is defined as $\sum_{a b} f_{a b c} \, f_{a b d} =\alg{c}_2 \, \delta_{c d}$, and it is related to trace of the quadratic Casimir in the adjoint representation, {\it i.e.} $\sum_a [t_a, [t_a,x]] = \alg{c}_2 \, x$, $\forall \, x \in \alg{g}$. The Killing form is nothing else but $\kappa_{c d} = \sum_{a b}f_{a b c} \, f_{a b d} = \alg{c}_2 \, \delta_{c d}$.} of $\alg{g}$ to be non vanishing. Then, there exists a transformation which reduces $r$ to a difference form. 
\end{BD2}

\begin{proof}
Without loss of generality, we can assume that the $r$-matrix will behave as
\begin{eqnarray}
r \sim \frac{\sum_a t_a \otimes t_a}{u_1 - u_2} + g(u_1,u_2)
\end{eqnarray}
near the origin. If it does not, then
\begin{eqnarray}
r \sim \frac{\xi(u_1) \sum_a t_a \otimes t_a}{u_1 - u_2} + g(u_1,u_2),
\end{eqnarray}
and the change of variables 
\begin{eqnarray}
\label{diue}
u = u(v), \qquad u'(v) = \xi(u(v))
\end{eqnarray}
will make the residue equal to $1$. In fact, near $v_1 = v_2$, one has
\begin{eqnarray}
\frac{\xi(u_1)}{u_1 - u_2} \sim \frac{\xi(u_1)}{u(v_2) + u'(v_2)(v_1 - v_2) - u(v_2)} \sim \frac{1}{v_1 - v_2}
\end{eqnarray}
due to (\ref{diue}).
The function $g$ can be taken to be holomorphic in a sufficiently small neighbourhood of the origin. 

Expanding (\ref{ndiffYBE}) near the point $u_2 = u_3$, we get
\begin{eqnarray}
&&[r_{12}(u_1,u_2), r_{13}(u_1,u_2)] + [r_{12}(u_1,u_2) + r_{13}(u_1,u_2), g_{23}(u_2,u_2)] +\nonumber \\
&&\qquad \Bigg[\sum_{a b} r_{a b}(u_1,u_2) t_{a} \otimes \Delta(t_b), 1 \otimes \frac{\sum_c t_c \otimes t_c}{u_2 - u_3}\Bigg] + \Bigg[\partial_{u_2} r_{12}(u_1,u_2), 1 \otimes \sum_a t_a \otimes t_a\Bigg]=0\nonumber
\end{eqnarray}
where $\Delta(t_a) = t_a \otimes 1 + 1 \otimes t_a$ coincides with the {\it trivial coproduct} on $\alg{g}$, when $\alg{g}$ is regarded as a {\it bialgebra} \cite{Chari:1994pz}. However, the third term cancels out because $\sum_a t_a \otimes t_a$ is the quadratic Casimir of $\alg{g}\otimes \alg{g}$, hence $[\Delta(t_b),\sum_a t_a \otimes t_a]=0$.

Now we apply the commutator map $x \otimes y \to [x,y]$
to the spaces $2$ and $3$ in the above equation, and use the fact that the dual Coxeter number is non-zero (and equal to 1 if we use appropriate conventions). After using the Jacobi identity, we get
\begin{eqnarray}
\sum_{a b c d}r_{a b}(u_1,u_2) r_{c d}(u_1,u_2) [t_a,t_c]\otimes [t_b,t_d]+[r(u_1,u_2), 1 \otimes h(u_2)] + \partial_{u_2} r(u_1,u_2)=0,
\end{eqnarray}
where $h(u) \equiv g_{a b}(u,u) [t_a, t_b]$.

We can repeat the same process on the variables $u_1$ and $u_2$, and the spaces $1$ and $2$. We obtain
\begin{eqnarray}
\sum_{a b c d}r_{a b}(u_1,u_3) r_{c d}(u_1,u_3) [t_a,t_c]\otimes [t_b,t_d]+[h(u_1) \otimes 1, r(u_1,u_3)] - \partial_{u_1} r(u_1,u_3)=0. 
\end{eqnarray}

In total, 
\begin{eqnarray}
\label{total}
\partial_{u_1} r(u_1,u_2) + \partial_{u_2} r(u_1,u_2) = [h(u_1) \otimes 1 + 1 \otimes h(u_2), r(u_1,u_2)].
\end{eqnarray}
Define an invertible map $\psi(u)$ acting on $\alg{g}$, such that 
\begin{eqnarray}
\frac{d}{du} \psi(u)[x] = \Big[h(u), \psi(u)[x]\Big] \qquad \forall \, \, x \in \alg{g},
\end{eqnarray}
and set
\begin{eqnarray}
\widehat{r}(u_1,u_2) = [\psi^{-1}(u_1) \otimes \psi^{-1}(u_2)] r(u_1,u_2).  
\end{eqnarray}
Since
\begin{eqnarray}
\frac{d}{du} \psi^{-1}(u) = - \psi^{-1} (u) \bigg(\frac{d}{du} \psi(u)\bigg) \psi^{-1} (u),
\end{eqnarray}
(\ref{total}) becomes
\begin{eqnarray}
\label{total1}
\partial_{u_1} \widehat{r}(u_1,u_2) + \partial_{u_2} \widehat{r}(u_1,u_2) = 0,
\end{eqnarray}
which shows that $\widehat{r}$ is of difference form.
\end{proof}

\smallskip

The importance of the two theorems above resides not only in their powerful classification of the possible classical integrable structures associated to simple Lie algebras, but also in how this structure turns out to determine quite uniquely the possible quantisations one can extract. This poses strong constraints on the possible types of infinite-dimensional quantum groups, thereby restricting the classes of quantum integrable systems one can ultimately realise.

Mathematically, the quantisation procedure involves the concept of {\it Lie bialgebras} and the so-called {\it Manin triples} (see for example \cite{EtingofSchiffman} and references therein). The term {\it quantisation} incorporates the meaning of completing the classical algebraic structure to a quantum group, or, equivalently, obtaining from a classical $r$-matrix a solution to the {\it quantum Yang-Baxter Equation}
\begin{eqnarray}
R_{12} \, R_{13} \, R_{23} \, = \, R_{23} \, R_{13} \, R_{12}, \qquad R_{ij} \sim 1 \otimes 1 \, + i \, \hbar \, r_{ij} + {\cal{O}} (\hbar^2).
\end{eqnarray}
The quantisation of the Sklyanin exchange relations is attained by simply {\it ``completing the $\hbar$ series"} into the famous {\it RTT} relations, which will be discussed at length during this school:
\begin{eqnarray}
\label{TEvge2}
\widehat{T}_1(u) \widehat{T}_2(u') R(u-u') \, = \, R(u-u') \widehat{T}_2(u') \widehat{T}_1(u), \qquad \widehat{T}(u) = T(u) + {\cal{O}}(\hbar)
\end{eqnarray} 
where the quantum monodromy $\widehat{T}$ is now understood as the normal-ordering of the classical product integral expression. We can see that (\ref{TEvge2}) tends to (\ref{TEvge}) for $\hbar \to 0$.

The associated quantum groups emerging from this quantisation process are then classified as {\it elliptic quantum groups} ($dim(\Gamma)=2$), {\it quantum affine algebras} ($dim(\Gamma)=1$), and {\it Yangians} ($\Gamma=\{0\}$, respectively)\footnote{A theorem \cite{Chari:1994pz} says that, if $r^3 = 0$, then $R=e^r$ solves the Yang-Baxter equation. In general, $R$ is a very complicated expression, which is the correspondent of the highly non-trivial procedure of quantising a classical Lie bialgebra \cite{Chari:1994pz,EtingofSchiffman}.}. We refer to the lectures by F. Loebbert at this school \cite{Florian} for more details on this.

This is then indeed a mathematical framework for transitioning from the classical to the quantum regime of the physics:
\begin{eqnarray}
\label{qcomm}
\{ A,B\} = \lim_{\hbar \to 0} \, \frac{[A,B]}{i \, \hbar}.
\end{eqnarray}
One could say that for integrable systems one has an explicit exact formula for the r.h.s. of (\ref{qcomm}) {\it as a function of $\hbar$}. In a sense, the Sklyanin exchange relations are the best starting point wherefrom to quantise the theory, in a way that keeps the integrable structure manifest at every step.
 
\subsection{Analytic properties}

In this section, we  discuss the analytic properties of the classical $r$-matrix as a function of the complex spectral parameters.

\bigskip

{\it \underline{Example}}

\medskip

\begin{itemize}

\item A convenient way of displaying the connection between the classical $r$-matrix and the associated quantum group is the case of Yangians. Let us consider the so-called {\it Yang's $r$-matrix}:
\begin{eqnarray}
\label{Yang}
r=\kappa \frac{{\cal{C_\otimes}}}{u_2 - u_1}.
\end{eqnarray}
This turns out to be the prototypical rational solution of the CYBE. Indeed, by definition of the Casimir ${\cal{C}_\otimes}$, one has $[{\cal{C}_\otimes},t^a \otimes \mathbbmss{1} + \mathbbmss{1} \otimes t^a]=0 \, \, \forall \, a$, and one can easily prove that (\ref{Yang}) solves the CYBE. 

This classical $r$-matrix is the one relevant for the non-linear Schr\"odinger model of example (\ref{NLS}), as it was proven by Sklyanin \cite{Evgeny}. Using this fundamental result, it is then an easy exercise to show that, combining (\ref{Yang}) and (\ref{TEvge}) with (\ref{TNLS}), one obtains in particular that $\{a(u),a(u')\}=0$, which was used in footnote \ref{fooNLS}. 

As a matter of fact, upon quantisation the NLS model is found to conserve the particle number, and in each sector of the Fock space with fixed number of particles it reduces to a quantum mechanical problem with mutual delta-function interactions \cite{Evgeny,Thacker}. This was exactly the context where Yang was working \cite{Yang} when he came across the solution (\ref{Yang}). In \cite{Evgeny}, Sklyanin went on to demonstrate that normal-ordering effects quantise the classical $r$-matrix (\ref{Yang}) into the canonical Yangian $R$-matrix (in suitable units):
$$
R = \mathbbmss{1} \otimes \mathbbmss{1} + \frac{i \kappa}{u_2 - u_1} \, {\cal{C}_\otimes},
$$
solution of the quantum Yang-Baxter equation.

\medskip

One can expand the classical $r$-matrix (\ref{Yang}) as follows:
\begin{eqnarray}
\label{Tayl}
&&\frac{r}{\kappa}= \frac{{\cal{C}_\otimes}}{u_2 - u_1} = \frac{\sum_a t_a \otimes t_a}{u_2 - u_1}=\sum_{a}\sum_{n\geq 0} t_a u_1^n \otimes t_a u_2^{-n-1}=\sum_{a}\sum_{n\geq 0} t_{a, n} \otimes t_{a,-n-1},\nonumber
\end{eqnarray}
where we have assumed $|\frac{u_1}{u_2}|<1$ for definiteness. Now we are capable of attributing the dependence on the spectral parameter $u_1$ (respectively, $u_2$) to generators in the first (respectively, second) space. This allows us to interpret the formula (\ref{Yang}) as the representation of an $r$-matrix, which is an abstract object living in the tensor product $\mathcal A_{u_1}[\mathfrak g] \otimes \mathcal A_{u_2}[\mathfrak g]$ of two copies of a larger algebra $\mathcal A_{u}[\mathfrak g]$ constructed out of $\mathfrak g$. The assignement 
\begin{eqnarray}
\label{assign}
t_{a, n} \, = \, u^n \, t_a
\end{eqnarray} 
in (\ref{Tayl}) produces
\begin{eqnarray}
\label{loo}
[t_{a, m} , t_{b, n}] = \,\sum_{c} f_{a b c} \, t_{c, m+n}
\end{eqnarray}
in terms of the structure constants $f_{a b c}$. The relations (\ref{loo}) then identify in this case the algebra $\mathcal A_{u}[\mathfrak g]$ as the {\it loop algebra}\footnote{A generalisation of the loop algebra is the so-called {\it affine Kac-Moody algebra} $\hat{\alg{g}}$, associated to a finite-dimensional Lie algebra $\alg{g}$. To obtain such a generalisation, one allows for a non-trivial central extension $c$. If we denote the generators of the affine Kac-Moody algebra as $s_{a,n} \equiv s_a\otimes v^n$ in terms of a formal parameter $v$, we can then write the defining relations of $\hat{\alg{g}}$ as 
\begin{eqnarray}
\label{aff}
[s_a\otimes v^n, s_b \otimes v^m]=[s_a,s_b]\otimes v^{n+m} + (s_a,s_b) \, n\, \delta_{n,-m} \, c, 
\end{eqnarray}
with $(,)$ the scalar product induced on $\alg{g}$ by the Killing form. One usually adjoins a derivation $d$ to the algebra:
\begin{eqnarray}
[d, s_a\otimes v^n] = n \, s_a\otimes v^n, \qquad d \equiv v \frac{d}{dv}, 
\end{eqnarray}
in order to remove a root-degeneracy (see {\it e.g.} \cite{Goddard:1986bp}).} ${\cal{L}}_u [\alg{g}]$ associated to $\alg{g}$.

One can then make sense of the operation of {\it abstracting} the relations (\ref{loo}) away from the specific representation (\ref{assign}) they are originally seen to emerge from. Using solely these commutation relations, one can then verify that the abstract formal expression\footnote{There are subtleties one needs to be careful about, regarding the convergence of formal series such as (\ref{formale}). An appropriate setting where to discuss such issues is typically provided by the so-called {\it $p$-adic topology} and {\it Poincar{\'e}-Birkhoff-Witt} bases.} 
\begin{eqnarray}
\label{formale}
r=\sum_a \sum_{n\geq 0} t_{a, n} \otimes t_{a,-n-1}
\end{eqnarray}
provides a consistent classical $r$-matrix independently of specific representations of (\ref{loo}). In turn, the universal enveloping algebra $U\big({\cal{L}}_u [\alg{g}]\big)$ of the loop algebra  ${\cal{L}}_u [\alg{g}]$ is nothing else but the classical limit of the Yangian $\Ya{\alg{g}}$:
\begin{eqnarray}
\Ya{\alg{g}} \to U\big({\cal{L}}_u [\alg{g}]\big) \qquad \mbox{as} \, \, \, \, \hbar \to 0.
\end{eqnarray}

\end{itemize}

\medskip

\newtheorem*{Pavel}{Theorem}

\begin{Pavel}
\label{PavelE}
The spans of the generators appearing separately on each factor of $r$ must form two Lie subalgebras of $\alg{g}$. 
\end{Pavel}

\begin{proof}
By writing $r = \sum_{ab} r_{a b} (u) \, z^a \otimes z^b$, with the $z$'s being a subset of the $t$'s, one has that near the pole at $u_1 = u_2$ the CYBE reduces to $$\sum_{a b c d}\frac{c_{a b}(u_1)}{u_1 - u_2} \, r_{c d} (u_1 - u_3) \, ([z_a , z_c]\otimes z_b \otimes z_d + z_a \otimes [z_b , z_c]\otimes z_d)=0,$$ with some function $c_{a b}(u_1)$. This implies $$[z_a , z_c] =\sum_d  f_{a c d} z_d$$ for a subset of the structure constants $f_{a c d}$. In \cite{BelavinDrinfeld1}, the Jacobi identity is shown, which proves that the two spans discussed above form Lie subalgebras of $\alg{g}$. 
\end{proof}

These two span subalgebras together with $\alg{g}$ form the so-called {\it Manin triple}. Characterisation of such a triple is an essential pre-requisite to identify the actual algebra, the quantum group is going to be built upon. In fact, suppose we had started with an $r$-matrix, solution of the classical Yang-Baxter equation, such that, however, none of the three requirements $(i)-(iii)$ in the  Belavin-Drinfeld I theorem held. In particular, this could be because we have identified that $r \in \alg{g} \otimes \alg{g}$, but some of the basis element $t_a$ never actually appear in $r$, causing the determinant $\det r_{a b}$ to vanish by a row of zeroes. The theorem we have just proven reassures us that we can always find a subalgebra $\alg{g}'$ of $\alg{g}$ such that $r \in \alg{g}' \otimes \alg{g}'$. For this restriction, $r$ has now a chance of being non-degenerate.  

\smallskip

Let us finally mention that classical $r$-matrices are core objects in the theory of quantum groups and deformation quantisation, and play a special role in the study of the so-called {\it Drinfeld double} (cf. F. Loebbert's lectures at this school \cite{Florian}). 

\section{Solitons}

In this section, we discuss the soliton solutions of integrable classical field theories. By way of introduction, let us report what is probably the archetype of solitons, contained in John Scott Russell's famous report of an event occurred at the Union Canal, Scotland, in 1834 \cite{Scott}: 

\medskip

{\it``I was observing the motion of a boat which was rapidly drawn along a narrow channel by a pair of horses, when the boat suddenly stopped - not so the mass of water in the channel which it had put in motion; it accumulated round the prow of the vessel in a state of violent agitation, then suddenly leaving it behind, rolled forward with great velocity, assuming the form of a large \underline{solitary elevation}, a rounded, smooth and well-defined heap of water, which continued its course along the channel apparently \underline{without change of form or diminution of} \underline{speed}".}

\medskip

In the language of integrable systems, we can model this situation using the famous Korteweg - de Vries (KdV) equation \cite{KdV1,KdV2} for a wave profile $\phi$ in shallow water\footnote{This form of the equation gets mapped onto the one presented in S. Negro's lectures at this school, upon the identifications $\phi = - U, t = - t_3, x = w$, and compactification of $x$ to $[0,2 \pi]$.}:
\begin{eqnarray}
\label{KdV}
\partial_t \, \phi \, + \, \partial_x^3 \, \phi \, - \, 6 \, \phi \, \partial_x \phi \, =0.
\end{eqnarray}
The KdV equation (\ref{KdV}) admits, as a particular solution\footnote{A more general solution is obtained in terms of the Jacobi elliptic cosine function `cn', hence the name {\it cnoidal wave}.}, a travelling soliton parameterised by two arbitrary real constants $x_0$ and $\upsilon$:
\begin{eqnarray}
\label{upsilon}
\phi \, = \, - \frac{\upsilon}{2} \, \mbox{sech}^2 \, \Bigg[\frac{\sqrt{\upsilon}}{2} \, (x - \upsilon \, t + x_0)\Bigg].
\end{eqnarray}

The Lax pair for the KdV equation is given by 
\begin{eqnarray}
L = \begin{pmatrix}0&-1\\u - \phi &0\end{pmatrix}, \qquad M = \begin{pmatrix}-\phi_x & -4u - 2 \phi\\ 4 u^2 - 2 u \, \phi \, + \phi_{xx} - 2 \phi^2 &\phi_x\end{pmatrix}.
\end{eqnarray}
The extra conservation laws of the integrable hierarchy prevent the waves from loosing their profiles throughout the time evolution. We will now discuss a very general method for solving integrable equations and find soliton solutions in a wide variety of cases. 

\subsection{The classical inverse scattering method}

Let us show how Gardner, Green, Kruskal and Miura solved the KdV equation, with a method which since became a standard procedure for integrable partial differential equations \cite{GardnerGreeneKruskalMiura}. As we pointed out in the Introduction, this method, dubbed of the {\it (classical) inverse scattering}, was afterwards adapted to quantum theories by the Leningrad school, and still to this day represents a paradigmatic approach to the quantisation of integrable systems. In the following section, we will show a more complicated application of the inverse scattering method to the Sine-Gordon equation.

\medskip

The main feature that emerged from early numerical calculations performed on the KdV equation was that there are solutions describing multiple propagating profiles, which nevertheless scatter off each other preserving their individual shapes through the process (see Figure \ref{fig:solitons}). 
\begin{figure}[t]
  \centering
  \begin{tikzpicture}
\begin{scope}[yshift=2.5cm]
\draw[-] (-6cm,0) to (6cm,0);
\draw[-] (0,-0.5cm) to (0,3cm);
\draw[scale=1,domain=-2:2,smooth,variable=\x, color=magenta,thick] plot ({0.8*\x-3},{2*exp(-3*\x*\x)});
\draw[scale=1,domain=-2:2,smooth,variable=\x, color=cyan,thick] plot ({0.8*\x+3},{2*exp(-1.2*\x*\x-3	.3*\x*\x*\x*\x)});
%
\node[color=magenta] at (-2.3,2.3) {$v_1 \longrightarrow$};
\node[color=cyan] at (2.5,2.4) {$\leftarrow v_2$};
\node[box] at (4.8,3.3) {$T\approx -\infty$};
\node[] at (6, 0.3) {$x$};
\node[] at (0.3, 3) {$\phi$};
\end{scope}
\begin{scope}[yshift=-2.5cm]
\draw[-] (-6cm,0) to (6cm,0);
\draw[-] (0,-0.5cm) to (0,3cm);
\draw[scale=1,domain=-2:2,smooth,variable=\x, color=magenta,thick] plot ({-0.8*\x+3},{2*exp(-3*\x*\x)});
\draw[scale=1,domain=-2:2,smooth,variable=\x, color=cyan,thick] plot ({-0.8*\x-3},{2*exp(-1.2*\x*\x-3	.3*\x*\x*\x*\x)});
%
\node[color=magenta] at (3.7,2.3) {$v_1 \longrightarrow$};
\node[color=cyan] at (-3.5,2.4) {$\leftarrow v_2$};
\node[box] at (4.8,3.3) {$T\approx +\infty$};
\node[] at (6, 0.3) {$x$};
\node[] at (0.3, 3) {$\phi$};
\end{scope}
\end{tikzpicture}
  \caption{Scattering of solitons keeping their shapes and velocities.}
  \label{fig:solitons}
\end{figure}
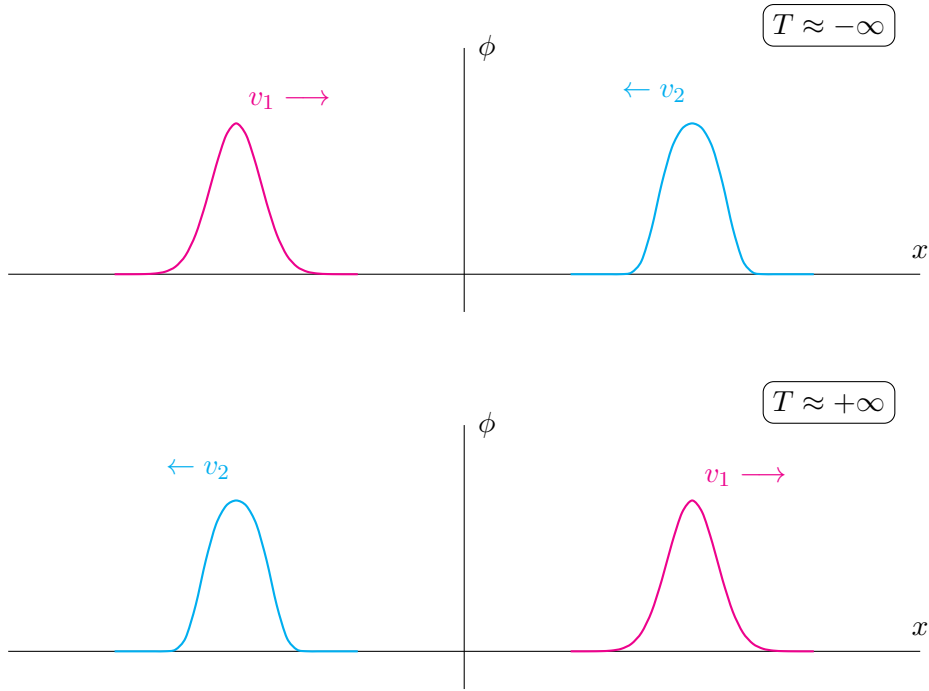
This is quite surprising for a non-linear equation as the one at hand, and it is due to a perfect competition between the non-linearity $\phi \, \partial_x \phi$ (trying to concentrate the profile) and the dispersion $\partial_x^3 \phi$ (trying to spread the profile). It also shows, in a way, how integrability is capable of restoring some features, which might be thought as rather pertaining to a linear behaviour, into a highly non-linear system.   

Gardner {\it et al.} in {\it op.~cit.} consider the auxiliary Schr\"odinger problem
\begin{eqnarray}
\label{silio}
\partial_x^2 \, \psi = \Big(\phi - u(t)\Big) \psi,
\end{eqnarray}
with $\phi$ satisfying (\ref{KdV}). Eq. (\ref{silio}) is equivalent to the first equation of our auxiliary linear problem (\ref{aux}), namely $\partial_x \Psi = L \Psi$, after one takes a further $\partial_x$ on the latter, and then projects it onto the first vector component. Solving for $\phi$ in (\ref{silio}) for $\psi$ not identically zero, and substituting back into (\ref{KdV}) gives
\begin{eqnarray}
\label{ausilio}
(\partial_t u) \, \psi^2 + \, [\psi \, Q_x - \, \psi_x Q]_x  =0, \qquad Q = \partial_t \, \psi \, + \, \partial_x^3 \, \psi \, - \, 3(\phi+u) \, \partial_x \psi.
\end{eqnarray}

One can see that, if $\psi$ vanishes sufficiently fast at $|x| \to \infty$, integrating the first equation in (\ref{ausilio}) on the whole real line implies $\partial_t u = 0$, hence $u$ is a constant {\it spectral parameter}. This means that we are effectively solving for the normalisable part of the spectral problem $\partial_x^2 \, \psi = (\phi -u) \psi$. It also means that we are left with solving
\begin{eqnarray}
\label{ausiliom}
[\psi \, Q_x - \, \psi_x Q]_x  = 0, \qquad \mbox{i.e.} \qquad \psi \, Q_{xx} = Q \, \psi_{xx},
\end{eqnarray}
from (\ref{ausilio}). It is then straightforward to check that, differentiating the equation
\begin{eqnarray}
\label{KdVmq}
Q(x,t) = C(t) \, \psi + D(t) \, \psi \int^x \frac{dx}{\psi^2} 
\end{eqnarray}
twice w.r.t. $x$, for two arbitrary integration constants $C(t)$ and $D(t)$, and re-using (\ref{KdVmq}) once, one obtains (\ref{ausiliom}).

At this point, we assume that $\phi$ vanishes at spatial infinity for any given time. 
\begin{itemize}
\item The normalisable $\psi$ modes, for $C=D=0$, satisfy
\begin{eqnarray}
0 = Q(x,t) \to \partial_t \, \psi \, + \, \partial_x^3 \, \psi \, - \, 3 u \, \partial_x \psi \qquad \mbox{at} \, \, |x| \to \infty
\end{eqnarray}
which is solved by the bound states
\begin{eqnarray}
\label{cnkn}
\psi_n \to c_n \, e^{\pm 4 k_n^3 t \mp k_n x} \qquad \mbox{at} \, \, x \to \pm \infty, \qquad k_n = \sqrt{- u_n}, \, \, \, u_n < 0, 
\end{eqnarray} 
expected to form the discrete part of the auxiliary spectral problem.

\item We can extend our problem to the non-normalisable modes with a wave-like behaviour at spatial infinity, choosing $u$ to be a constant. For this, we can first go back to (\ref{silio}) and deduce for instance, for $k^2 = u>0$,
\begin{eqnarray}
&&\psi \longrightarrow e^{- i k x} + b \, e^{i k x}, \qquad x \to \infty, \nonumber\\
&&\psi \longrightarrow a \, e^{-i k x}, \qquad x \to - \infty.
\end{eqnarray} 
Solutions which are asymptotically plane waves at $x \to \pm \infty$ are called {\it Jost solutions}. Plugging this into (\ref{KdVmq}), one finds as a solution $D=0$, $C= 4i k^3$, and {\it scattering data} $a,b$ determined as
\begin{eqnarray}
\label{ab}
a(k,t) = a(k,0), \qquad b(k,t) = b(k,0) \, e^{8 i k^3 t}.
\end{eqnarray}
This is called the {\it direct scattering problem}.

\end{itemize}

Combining together all this information turns out to be \underline{sufficient} to reconstruct $\phi$. This means that we can reconstruct the potential $\phi$ in the auxiliary Schr\"odinger problem (\ref{silio}) from the scattering data. This is the {\it inverse scattering problem}. In fact, if $K(x,y)$, for $y \geq x$, is a solution of the {\it Gel'fand-Levitan-Marchenko} equation \cite{GelfandLevitan,Marchenko}
\begin{eqnarray}
\label{GLM}
K(x,y)+B(x+y)+\int_{-\infty}^x dz \, K(x,z) \, B(y+z) = 0, 
\end{eqnarray}
where
\begin{eqnarray}
\label{ref}
B(x) = \frac{1}{2 \pi} \int_{-\infty}^\infty dk \, b(k) \, e^{i k x} + \sum_n c_n^2 e^{- 8 k_n^3 t} \, e^{k_n x}
\end{eqnarray}
in terms of the coefficients $b$ in (\ref{ab}) and $c_n,k_n$ in (\ref{cnkn}), then one has
\begin{eqnarray}
\label{GLM1}
\phi = 2 \frac{d}{dx} K(x,x).
\end{eqnarray} 

The theory behind the Gel'fand-Levitan-Marchenko equation is deeply rooted in the technology which allows one to reconstruct the potential of a given Schr\"odinger problem from the knowledge of its reflection and transmission coefficients, which both feature in (\ref{ref}). The classical inverse scattering method is also regarded as a generalisation of the Fourier transform to non-linear problems. Let us sketch here an argument\footnote{A rigorous derivation of the Gel'fand-Levitan-Marchenko equation is beyond the scope of these lectures. For a wider context, the interested reader can for instance consult \cite{Calogero,Nettel}. Here, we found it convenient to follow a discussion by Terry Tao \cite{Tao}.} that motivates formulae (\ref{GLM})-(\ref{GLM1}) in a simplified case. Consider the spectral problem
\begin{eqnarray}
\label{ansa}
- \frac{\partial^2}{\partial t^2} w + \frac{\partial^2}{\partial x^2} w = V(x) w,
\end{eqnarray}
where $V(x)$ has a compact support $[-R,R]$ in the spatial direction $x$. This means that, in the regions $x<-R$ and $x>R$, $w$ satisfies the free wave equation, hence we can write
\begin{eqnarray}
&&w = f_-(x-t) + g_-(x+t), \qquad x <-R, \nonumber\\
&&w = f_+(x-t) + g_+(x+t), \qquad x > R.
\end{eqnarray}
Consider now two different solutions, characterised by the following asymptotics:
\begin{itemize}
\item {\it Case 1}
\begin{eqnarray}
&&w = \delta(x-t), \qquad t <<-R, \nonumber\\
&&w =  g_-(x+t) + f_+(x-t), \qquad t >> R.
\end{eqnarray}
For finite speed of propagation, this solution vanishes for $x>t$.
\item {\it Case }2
\begin{eqnarray}
&&w = \delta(x-t) + g_+(x+t), \qquad t<<-R, \nonumber\\
&&w =  f_+(x-t), \qquad t >> R.
\end{eqnarray}
\end{itemize}
In both cases (which, in a sense, can be thought of as being dual to each other), the functions $f_\pm$ and $g_\pm$ appearing will be pictured as some sufficiently localised profiles. Let us give a special name to the function $g_-$, {\it i.e.}
\begin{eqnarray}
g_-(z) \equiv B(z) \qquad \mbox{\it ``scattering data".} 
\end{eqnarray}
Let us focus on case 2, and make the ansatz 
\begin{eqnarray}
\label{u}
w = \delta(x-t) + K(x,t) \Theta(x-t)
\end{eqnarray}
for the full solution to (\ref{ansa}), with $\Theta$ the Heaviside step function. Take $K$ to vanish for $x<-R$. Plugging this ansatz back into (\ref{ansa}) and collecting the terms proportional to $\delta(x-t)$, one gets
\begin{eqnarray}
V(x) = 2 \frac{d}{dx}K(x,x).
\end{eqnarray}
One then notices that, if $w(x,t)$ solves (\ref{ansa}), so does $w(x,-t+s)$ for an arbitrary constant shift $s$. Therefore, also
\begin{eqnarray}
w(x,t) + \int_{-\infty}^\infty ds \, B(s) \, w(x,-t+s)
\end{eqnarray} 
does. If we use the ansatz (\ref{u}), this means that
\begin{eqnarray}
\delta(x-t) + K(x,t)\Theta(x-t)+B(x+t)+\int_{-\infty}^x ds \, B(s+t) \, K(x,s)
\end{eqnarray} 
solves (\ref{ansa}), and it coincides with $\delta(x-t) + B(x+t)$ when $x<-R$. Therefore, this solution corresponds to case 1 above, hence it must vanish for $x>t$. This in turn implies
\begin{eqnarray}
K(x,t)+B(x+t)+\int_{-\infty}^x ds \, B(s+t) \, K(x,s)=0.
\end{eqnarray}
The form we use in the case of the KdV equation basically involves the Fourier transform of the procedure we have just sketched.

As an example, the single-soliton solution of the KdV equation is obtained from the above procedure in the case when $b=0$, and there is only one discrete eigenvalue $u$. 
In this situation, one simply has 
\begin{eqnarray}
B(x) = \gamma \, e^{k x}, \qquad \gamma \equiv c^2 \, e^{- 8 k^3 t}
\end{eqnarray}
and it is therefore convenient to make an ansatz for $K$ of the form
\begin{eqnarray}
K(x,y) = K(x) \, e^{k y}.
\end{eqnarray}
Eq. (\ref{GLM}) then becomes easy to solve after a simple integration for $k>0$: 
\begin{eqnarray}
K(x)=-\frac{2 \gamma \, k \, e^{k x}}{\gamma e^{2 k x} + 2 k}.
\end{eqnarray}
We immediately get from (\ref{GLM1}) that
\begin{eqnarray}
\phi = - \frac{16 \, c^2 \, k^3 \, e^{2 k \big(- 4 k^2 t + x\big)}}{\big[c^2 \, e^{2 k \big(- 4 k^2 t + x\big)} + 2 k\big]^2}.
\end{eqnarray}
For the choice $c^2 = 2 k$, we obtain
\begin{eqnarray}
\phi \, = \, - 2 \, k^2 \, \mbox{sech}^2 \Big[k \big(-4 k^2 t + x\big)\Big],
\end{eqnarray}
which coincides with (\ref{upsilon}) for $\upsilon = 4 k^2$.
For further detail and a complete exposition, we recommend the lecture notes \cite{Dunajski}.

\subsection{The Sine-Gordon equation and Jost solutions}

We will now apply the inverse scattering method to the more complicated example of the Sine-Gordon theory (\ref{SG}). In this section, we will follow \cite{BabelonBernardTalon}.

Let us look for solutions of (\ref{SG}) behaving at infinity like
\begin{eqnarray}
\label{bc}
\phi \to 0 \qquad \mbox{as} \, \, x \to - \infty, \qquad \qquad \phi \to \frac{q \pi}{\beta} \qquad \mbox{as} \, \, x \to + \infty.
\end{eqnarray} 
The quantity $q$ is called the {\it topological charge}. We will restrict to the case where $q$ is an integer, which is compatible with the equation. It can be calculated as
\begin{eqnarray}
q = \frac{\beta}{\pi} \int_{-\infty}^\infty dx \, \frac{\partial \phi}{\partial x} = \frac{\beta}{\pi} \Big[ \phi(\infty) - \phi(-\infty)\Big].
\end{eqnarray} 
If $q=1$ we speak of a soliton (or {\it kink}), while if $q=-1$ we speak of an anti-soliton (or {\it anti-kink}).

\bigskip

\noindent {\it Direct Problem}

\bigskip 

\noindent We begin by focusing on what is called the {\it direct scattering problem}, namely, 
the first equation in (\ref{aux}):
\begin{eqnarray}
\label{direct}
\partial_x \Psi \, = \, L \, \Psi, \qquad \Psi = \begin{pmatrix}\psi_1\\ \psi_2\end{pmatrix}. 
\end{eqnarray}
At infinity, the conditions (\ref{bc}) imply for (\ref{direct})
\begin{eqnarray}
\partial_x \Psi \, \to \, i k \, \sigma_1 \Psi \qquad \mbox{as} \, \, x \to - \infty, \qquad \qquad \partial_x \Psi \, \to \, i k e^{i q \pi} \, \sigma_1 \Psi \qquad \mbox{as} \, \, x \to + \infty, 
\end{eqnarray} 
where $\sigma_1$ is a Pauli matrix, and
\begin{eqnarray}
k = m u \, - \, \frac{m}{u}.
\end{eqnarray}
This means that the solution will behave as plane waves at infinity. Let us define 
\begin{eqnarray}
\bar{\Psi} = \begin{pmatrix}\bar{\psi}_1\\\bar{\psi}_2\end{pmatrix} = \begin{pmatrix}- i \psi_2^*\\i \psi_1^*\end{pmatrix}.
\end{eqnarray}
We will make the specific choice of two {\it Jost solutions}, determined by the following asymptotic behaviours;
\begin{equation}
\Psi_1 \, \to \, a \, \begin{pmatrix}1\\1\end{pmatrix} e^{i k x} - b \begin{pmatrix}1\\-1\end{pmatrix} e^{-i k x}\qquad x \to - \infty, \qquad \qquad \Psi_1 \, \to \, \begin{pmatrix}1\\e^{i q \pi}\end{pmatrix} e^{i k x} \qquad x \to + \infty, \nonumber
\end{equation}
\begin{equation}
\Psi_2 \to \begin{pmatrix}1\\-1\end{pmatrix} e^{- i k x} \qquad x \to - \infty, \qquad \Psi_2 \to -b^* \begin{pmatrix}1\\e^{i q \pi}\end{pmatrix} e^{i k x} + a \begin{pmatrix}e^{i q \pi}\\-1\end{pmatrix} e^{-i k x}\qquad x \to + \infty, \nonumber
\end{equation} 
\begin{eqnarray}
|a(u)|^2 \, + \, |b(u)|^2 \, = \, 1.
\end{eqnarray}
The last condition follows from considering that various Wronskians constructed with $\Psi_i$ and $\bar{\Psi}_i$ are independent of $x$. 

The two solutions $\Psi_1$ and $\Psi_2$ are independent as long as the Wronskian $\det |\Psi_1 \, \Psi_2|$ does not vanish, namely, as long as the {\it Jost function} $a(u)$ does not have zeroes. Let us first put ourselves away from any zero of $a(u)$, and list a few properties of the two independent Jost solutions (without proof).

\begin{itemize}

\item The Jost solutions $\Psi_1$ and $\Psi_2$, when regarded as functions of $u$, are analytic in the upper half plane $\mbox{Im} (u) > 0$. 

\item For a fixed $x$, the Jost solutions have the following asymptotics in $u$:
\begin{equation}
\Psi_1 \, \to \, e^{- i \frac{q \pi}{2}} \begin{pmatrix}e^{i \frac{\phi \beta}{2}}\\e^{-i \frac{\phi \beta}{2}}\end{pmatrix} e^{i k x} \qquad |u| \to \infty, \qquad \Psi_1 \, \to \, e^{i \frac{q \pi}{2}} \begin{pmatrix}e^{-i \frac{\phi \beta}{2}}\\e^{i \frac{\phi \beta}{2}}\end{pmatrix} e^{i k x} \qquad |u| \to 0, \nonumber
\end{equation}
\begin{equation}
\Psi_2 \, \to \, \begin{pmatrix}e^{i \frac{\phi \beta}{2}}\\-e^{-i \frac{\phi \beta}{2}}\end{pmatrix} e^{-i k x} \qquad |u| \to \infty, \qquad \Psi_2 \, \to \, \begin{pmatrix}e^{-i \frac{\phi \beta}{2}}\\-e^{i \frac{\phi \beta}{2}}\end{pmatrix} e^{-i k x} \qquad |u| \to 0. \nonumber
\end{equation} 

\item The Jost function $a(u)$ is analytic in the upper half plane $\mbox{Im} (u) > 0$ and satisfies
\begin{eqnarray}
&&a(u) \, \to \, e^{- i \frac{q \pi}{2}} \qquad |u| \to \infty, \qquad a(u) \, \to \, e^{i \frac{q \pi}{2}} \qquad \qquad |u| \to 0, \nonumber\\
&& a(- u) \, = \, e^{- i q \pi} \, a^*(u) \qquad \mbox{for real} \, \, u.
\end{eqnarray}

\end{itemize}  

The two functions $a(u)$ and $b(u)$, the zeroes $u_n$ of $a(u)$ and the proportionality constants $c_n$ between $\Psi_1$ and $\Psi_2$ at the values $u_n$ where the Jost solutions become dependent 
$$\Psi_2 (x,u_n) = c_n \, \Psi_1 (x,u_n),$$ 
altogether form the set of {\it scattering data}.

The time evolution of the scattering data can now be explicitly computed by substitution into the {\it second} equation in (\ref{aux}). One gets
\begin{eqnarray}
\label{data}
&&a(u,t) = a(u,0), \qquad b(u,t) = e^{2 i m \big(u + \frac{1}{u}\big) t} \, b(u,0), \qquad \mbox{hence}\nonumber\\
&& u_n (t) = u_n(0), \qquad c_n(t) = e^{-2 i m \big(u_n + \frac{1}{u_n}\big) t} \, c_n(0).
\end{eqnarray}  

\bigskip

\noindent {\it Inverse Problem}

\bigskip 

\noindent We now proceed to reconstruct the field $\phi$ from the solutions to the auxiliary linear problem. For this, we start by noticing that for real $u$ the Jost solutions satisfy 
\begin{eqnarray}
\frac{1}{a(u)} \, e^{i\frac{q \pi}{2} \sigma_3} \, \widehat{\Psi}_2 \, = \, \xi (u) \, \widehat{\Psi}_1 \, + i \, \bar{\widehat{\Psi}}_1,
\end{eqnarray}
where
\begin{eqnarray}
\xi (u) = - \frac{b^*(u)}{a^*(u)}, \qquad \widehat{\Psi}_1 = e^{i\frac{q \pi - \beta \phi}{2} \sigma_3} \Psi_1, \qquad \widehat{\Psi}_2 = e^{- i\frac{\beta \phi}{2} \sigma_3} \Psi_1.
\end{eqnarray}
This can be obtained by using that, if $\psi(x,u)$ is a solution to the direct problem (\ref{direct}), so is $\sigma_3 \, \psi(x,-u)$. One can then compare the asymptotic behaviours of $\bar{\Psi}_i(x,u)$ and $\sigma_3 \Psi_i(x,-u)$, where $\Psi_i$ are the two Jost solutions.

One can also prove that the (hatted) Jost solutions admit the following Fourier representation:
\begin{eqnarray}
\label{Four}
&&\widehat{\Psi}_1 = e^{i k x} \begin{pmatrix}1\\e^{i q \pi}\end{pmatrix} + \int_x^\infty dy \, \bigg(v_1 (x,y) + \frac{1}{u} w_1(x,y)\bigg) \, e^{i k y}, \nonumber\\
&&\widehat{\Psi}_2 = e^{-i k x} \begin{pmatrix}1\\-1\end{pmatrix} + \int_{-\infty}^x dy \, \bigg(v_2 (x,y) + \frac{1}{u} w_2(x,y)\bigg) \, e^{-i k y},
\end{eqnarray}
for $v_i(x,y)$ and $w_i(x,y)$ two-component vectors with sufficiently regular behaviour.

The importance of the kernels $v_i$ and $w_i$ is that the knowledge of their explicit form allows to reconstruct the matrix $L$ and therefore provides a solution for the original field $\phi$. One in fact has that, plugging these expansions back into the direct problem, 
\begin{eqnarray}
\label{what}
e^{2 i \beta \phi(x)} \, = \, \frac{i m + e^{i q \pi} w_{1,2}(x,x)}{i m + w_{1,1}(x,x)},
\end{eqnarray}
where $q$ is defined in (\ref{bc}) and $w_{i,j}$ is the component $j$ of the vector $w_i$. 

The final step of the process involves again the Gel'fand-Levitan-Marchenko equation for the kernels $v_i$ and $w_i$ appearing in the Fourier decomposition (\ref{Four}).

\newtheorem*{GLM}{Theorem (Gel'fand-Levitan-Marchenko)}

\begin{GLM}
\label{GLME}
The kernels $v_1,w_1$ in (\ref{Four}) satisfy the following integral equations:
\begin{eqnarray}
&&\bar{v}_1(x,y) \, = \, f_0(x+y) \, \begin{pmatrix}1\\e^{i q \pi}\end{pmatrix}+ \int_x^\infty dz \, \bigg(f_0(z+y) \, v_1 (x,z) + f_{-1}(z+y) w_1(x,z)\bigg), \nonumber\\
&&\bar{w}_1(x,y) \, = \, f_{-1}(x+y) \, \begin{pmatrix}1\\e^{i q \pi}\end{pmatrix} + \int_x^\infty dz \, \bigg(f_{-1}(z+y) \, v_1 (x,z) + f_{-2}(z+y) w_1(x,z)\bigg),\nonumber
\end{eqnarray}
with the scattering data (\ref{data}) entering these equations as
\begin{eqnarray}
f_i (x) = -\frac{m}{2 \pi i} \, \int_{-\infty}^{\infty} d u \, u^i \, e^{i k x} \, \xi(u) + m \sum_n e^{i k(u_n) x} u_n^i m_n, \qquad m_n = \frac{c_n}{a'(u_n)}.
\end{eqnarray}

\end{GLM}

The Gel'fand-Levitan-Marchenko equations greatly simplify under the assumption that $b=0$, corresponding to absence of reflection in the auxiliary linear problem. The integral equation for $v_1$ reduces to a linear system. The simplest possible assumption is that there is only one purely imaginary zero $u_1 = i \mu_1$ of $a(u)$, with $\mu_1<0$, which collapses the entire system to one linear equation with straightforward solution. Plugging this solution into (\ref{what}) returns the one-soliton solution
\begin{eqnarray}
e^{i \beta \phi} = \frac{1+X}{1-X}, \qquad X = i \, e^{- 2 m \big(\mu_1 + \frac{1}{\mu_1}\big) (x-x_0) + 2 m \big(\mu_1 - \frac{1}{\mu_1}\big)t}, \qquad \mu_1 \equiv - \sqrt{\frac{1+v}{1-v}}, 
\end{eqnarray}
which is a profile propagating with velocity $v=\frac{\mu_1^2 - 1}{\mu_1^2 + 1} \in [-1,1]$.  It is easy to see that the topological charge $q$ in (\ref{bc}) equals $1$ for this solution. Upon quantisation, solitons and anti-solitons become the elementary particles in quantum Sine-Gordon theory.
  
The other famous and slightly more complicated solution one obtains, is one that is normally quantised into a bound state of a soliton and an anti-soliton, and it accordingly has zero topological charge. This solution is called a {\it breather}. The profile for a (non-translating) breather is given by
\begin{eqnarray}
\phi = \frac{2}{\beta} \arctan \frac{\sqrt{1-\omega^2} \, \cos (4 m \omega t)}{\omega \, \cosh \big(4 m \sqrt{1-\omega^2} x\big)}, \qquad \omega \in [0,1].
\end{eqnarray}

\section{Conclusions}

We hope that these lectures have stimulated the curiosity of the many young researchers present at this school to delve into the problematics of integrable systems, and have prepared the ground for the following lectures, where, in particular, the quantum version of integrability will be presented. 

This field is constantly growing, attracting representatives of different communities with interests ranging from mathematics to mathematical physics and high-energy physics. We are sure that the new generation of physicists and mathematicians present here in Durham will accomplish great progress in all these avenues of investigation. 

\section{Acknowledgements}

It is a pleasure to thank the Department of Mathematical Sciences at
Durham University, the GATIS network (in particular Zoltan Bajnok, Patrick Dorey and Roberto Tateo, who where present in Durham, and Charlotte Kristjansen) and the organisers of the school (Alessandra Cagnazzo, Rouven Frassek, Alessandro Sfondrini, Istv\'an Sz\'ecs\'enyi and Stijn van Tongeren) for their support and fantastic work , and all the lecturers for close cooperation on the choice of material (Fedor Levkovich-Maslyuk), thorough proof-reading (Florian Loebbert) and for very valuable comments and extremely useful discussions (Diego Bombardelli, Stefano Negro, Stijn van Tongeren). A special thanks goes to A. Sfondrini for providing the figures displayed in this review.

The author also thanks Andrea Prinsloo and Paul Skerritt for a careful proof-reading of the manuscript and for very useful suggestions and comments. He thanks Maxim Zabzine and Pavel Etingof for the discussions that brought to delving into the subject of section 4. He thanks the EPSRC for funding under the First Grant project EP/K014412/1 ``Exotic quantum groups, Lie superalgebras and integrable systems", and the STFC for support under the Consolidated Grant project nr. ST/L000490/1 ``Fundamental Implications of Fields, Strings and Gravity". He acknowledges useful conversations with the participants of the ESF and STFC supported workshop ``Permutations and Gauge String duality" (STFC-4070083442, Queen Mary U. of London, July 2014). 

Finally, it is a pleasure to thank the students attending the school, for making it such an exciting event with their questions and suggestions, which were very much appreciated.

No data beyond those presented and cited in this work are needed to validate this study.


\begin{thebibliography}{99}

\bibitem{BabelonBernardTalon}
O.~Babelon, D.~Bernard and M.~Talon, ``Introduction to classical integrable systems,'' Cambridge University Press, 2003.

\bibitem{FaddeTak}
  L.~D.~Faddeev and L.~A.~Takhtajan,
  ``Hamiltonian Methods in the Theory of Solitons,''
Berlin: Springer (1987)

\bibitem{Novikov:1984id}
  S.~Novikov, S.~V.~Manakov, L.~P.~Pitaevsky and V.~E.~Zakharov,
  ``Theory Of Solitons. The Inverse Scattering Method,''
  New York, Usa: Consultants Bureau (1984) 276 P. (Contemporary Soviet Mathematics)

\bibitem{Gleb}
G.~Arutyunov, ``Student Seminar: Classical and Quantum Integrable Systems'', lectures delivered at Utrecht Univ., 2006-07 [http:// www.staff.science.uu.nl/ $\sim$ aruty101/ teaching.htm]

\bibitem{ETH}
C.~Candu, M.~de~Leeuw, ``Introduction to Integrability'', lectures delivered at ETH, 2013, [http://www.itp.phys.ethz.ch/education/fs13/int]

\bibitem{GardnerGreeneKruskalMiura}
C.~S.~Gardner, J.~M.~Greene, M.~D.~Kruskal and R.~M.~Miura, ``Method for solving the Korteweg-deVries equation,'' Phys.\ Rev.\ Lett.\ {\bf 19} (1967) 1095.

\bibitem{ZS1}
  A.~B.~Shabat and V.~E.~Zakharov,
  ``A scheme for integrating the nonlinear equations of mathematical physics by the method of the inverse scattering problem. I,''
  Funct.\ Anal.\ Appl.\  {\bf 8} (1974) 226.
   
\bibitem{ZS2}
  A.~B.~Shabat and V.~E.~Zakharov,
  ``Integration of nonlinear equations of mathematical physics by the method of inverse scattering. II,''
  Funct.\ Anal.\ Appl.\  {\bf 13} (1979) 166.

\bibitem{Evgeny}
  E.~K.~Sklyanin,
  ``Quantum version of the method of inverse scattering problem,''
  J.\ Sov.\ Math.\  {\bf 19} (1982) 1546
   [{\it Zap.\ Nauchn.\ Semin.}  {\bf 95} (1980) 55].
  
\bibitem{FaddeevZakharov}
  L.~D.~Faddeev and V.~E.~Zakharov,
  ``Korteweg-de Vries equation: A Completely integrable Hamiltonian system,''
  Funct.\ Anal.\ Appl.\  {\bf 5} (1971) 280
   [{\it Funkt.\ Anal.\ Pril.}  {\bf 5N4} (1971) 18].

\bibitem{Drinfeld}
  V.~G.~Drinfeld,
  ``Quantum groups,''
  J.\ Sov.\ Math.\  {\bf 41} (1988) 898
   [{\it Zap.\ Nauchn.\ Semin.} {\bf 155} (1986) 18].
  
\bibitem{ZamolodchikovZamolodchikov}
  A.~B.~Zamolodchikov and A.~B.~Zamolodchikov,
  ``Factorized $S$-Matrices in Two-Dimensions as the Exact Solutions of Certain Relativistic Quantum Field Models,''
  Annals Phys.\  {\bf 120} (1979) 253.

\bibitem{Baxter}
  R.~J.~Baxter,
  ``Exactly solved models in statistical mechanics,'' Courier Corporation, 2007
  
\bibitem{Stijn}
  S.~J.~van Tongeren,
  ``Introduction to the thermodynamic Bethe ansatz,''
  J.\ Phys.\ A {\bf 49} (2016) no.32,  323005
    [arXiv:1606.02951 [hep-th]], A.~Cagnazzo, R.~Frassek, A.~Sfondrini,
I.~M.~Sz\'ecs\'enyi, S.~J.~van Tongeren 
editors. 

\bibitem{Ulam}
  S.~M.~Ulam,
  ``On combination of stochastic and deterministic processes - preliminary report,'' Bulletin of the American Mathematical Society, {\bf 53} (1947) 1120.
  
\bibitem{Ward:1985gz}
  R.~S.~Ward,
  ``Integrable and solvable systems, and relations among them,''
  Phil.\ Trans.\ Roy.\ Soc.\ Lond.\ A {\bf 315} (1985) 451.

\bibitem{Ablowitz:2003bv}
  M.~J.~Ablowitz, S.~Chakravarty and R.~G.~Halburd,
  ``Integrable systems and reductions of the self-dual Yang-Mills equations,''
  J.\ Math.\ Phys.\  {\bf 44} (2003) 3147.
  
\bibitem{Faddeev}
  L.~D.~Faddeev,
  ``How algebraic Bethe ansatz works for integrable model,'', Les Houches lectures, 1995 [hep-th/9605187].
  
    \bibitem{MironovEtAl}
  A.~Mironov, A.~Morozov, Y.~Zenkevich and A.~Zotov,
  ``Spectral Duality in Integrable Systems from AGT Conjecture,''
  JETP Lett.\  {\bf 97} (2013) 45
   [{\it Pisma Zh.\ Eksp.\ Teor.\ Fiz.}  {\bf 97} (2013) 49]
  [arXiv:1204.0913 [hep-th]].

\bibitem{LaxKepler}
  M.~Antonowicz and S.~Rauch-Wojciechowski,
  ``Lax representation for restricted flows of the KdV hierarchy and for the Kepler problem,''
  Phys. Lett. A {\bf 171} (1992) 303.
  
\bibitem{SZ}
  A.~B.~Shabat and V.~E.~Zakharov,
  ``Exact theory of two-dimensional self-focusing and one-dimensional self-modulation of waves in nonlinear media,''
  Soviet Physics JETP {\bf 34} (1972) 62.
  
\bibitem{T}
  L.~A.~Takhtajan,
  ``Hamiltonian systems connected with the Dirac equation,''
  Zapiski Nauchnykh Seminarov POMI {\bf 37} (1973) 66.
  
\bibitem{Maillet1}
  J.~M.~Maillet,
  ``Hamiltonian Structures for Integrable Classical Theories From Graded Kac-moody Algebras,''
  Phys.\ Lett.\ B {\bf 167} (1986) 401.
  
\bibitem{Maillet2}
  J.~M.~Maillet,
  ``New Integrable Canonical Structures in Two-dimensional Models,''
  Nucl.\ Phys.\ B {\bf 269} (1986) 54.

\bibitem{Niall}
  N.~J.~MacKay,
  ``On the classical origins of Yangian symmetry in integrable field theory,''
  Phys.\ Lett.\ B {\bf 281} (1992) 90
   [Phys.\ Lett.\ B {\bf 308} (1993) 444].
  
\bibitem{Fedor}
F.~Levkovich-Maslyuk, ``Lectures on the Bethe Ansatz," J.\ Phys.\ A {\bf 49} (2016) no.32,  323004
  [arXiv:1606.02950 [hep-th]], A.~Cagnazzo, R.~Frassek, A.~Sfondrini,
I.~M.~Sz\'ecs\'enyi, S.~J.~van Tongeren 
editors. 

\bibitem{BelavinDrinfeld1}
A.~A.~Belavin and V.~G.~Drinfeld,
``Solutions of the classical Yang-Baxter equation for simple Lie algebras'',
Funct. Anal. Appl. {\bf 16} (1982), 159

\bibitem{BelavinDrinfeld2}
A.~A.~Belavin and V.~G.~Drinfeld,
``Triangle equation for simple Lie algebras'',
Mathematical Physics Reviews (ed. Novikov at al.) Harwood, New York 
(1984), 93

\bibitem{BelavinDrinfeld3}
A.~A.~Belavin and V.~G.~Drinfeld,
``Classical Yang-Baxter equation for simple Lie algebras'',
Funct. Anal. Appl. {\bf 17} (1983), 220

\bibitem{Chari:1994pz}
  V.~Chari and A.~Pressley,
  ``A guide to quantum groups,''
  Cambridge, UK: Univ. Pr., 1994

\bibitem{EtingofSchiffman}
P.~Etingof and O.~Schiffman,
``Lectures on Quantum Groups,''
Lectures in Mathematical Physics, International Press, Boston, 1998 

\bibitem{Florian}
F.~Loebbert, ``Lectures on Yangian Symmetry", J.\ Phys.\ A {\bf 49} (2016) no.32,  323002
  [arXiv:1606.02947 [hep-th]], A.~Cagnazzo, R.~Frassek, A.~Sfondrini,
I.~M.~Sz\'ecs\'enyi, S.~J.~van Tongeren 
editors. 

\bibitem{Thacker}
  H.~B.~Thacker,
  ``Exact Integrability in Quantum Field Theory and Statistical Systems,''
  Rev.\ Mod.\ Phys.\  {\bf 53} (1981) 253.
  
\bibitem{Yang}
  C.~N.~Yang,
  ``Some exact results for the many body problems in one dimension with repulsive delta function interaction,''
  Phys.\ Rev.\ Lett.\  {\bf 19} (1967) 1312.
 
\bibitem{Goddard:1986bp}
  P.~Goddard and D.~I.~Olive,
  ``Kac-Moody and Virasoro Algebras in Relation to Quantum Physics,''
  Int.\ J.\ Mod.\ Phys.\ A {\bf 1} (1986) 303.
  
\bibitem{Scott}
J.~S.~Russell, ``Report on Waves," Report of the fourteenth meeting of the British Association for the Advancement of Science, York, 1844 (London 1845), 311

\bibitem{KdV1}
J.~Boussinesq, {\it ``Essai sur la theorie des eaux courantes, Memoires presentes par divers savants," Acad. des Sci. Inst. Nat. France}, XXIII, 1877 

\bibitem{KdV2}
D.~J.~Korteweg, G.~de Vries, ``On the Change of Form of Long Waves Advancing in a Rectangular Canal, and on a New Type of Long Stationary Waves", Philosophical Magazine 39 (240) (1895), 422

\bibitem{GelfandLevitan}
I.~M.~Gelʹfand and B.~M.~Levitan,
 ``On the determination of a differential equation from its spectral function,'' American Mathematical Society, 1955. 
 
\bibitem{Marchenko}
V.~A.~Marchenko, ``Sturm-Liouville operators and their applications," {\it Kiev Izdatel Naukova Dumka} 1, 1977.

\bibitem{Calogero}
F.~Calogero, A.~Degasperis, ``Spectral Transform and Solitons I,", North-Holland, 1982 

\bibitem{Nettel}
S.~Nettel, ``Wave Physics. Oscillations - Solitons - Chaos,", Springer, 2003

\bibitem{Tao}
T.~Tao, ``Israel Gelfand", entry at https://terrytao.wordpress.com/2009/10/07/israel-gelfand/$\#$more-2860

\bibitem{Dunajski}
M.~Dunajski, ``Integrable systems", University of Cambridge Lecture Notes, 2012 [http://www.damtp.cam.ac.uk/user/md327/teaching.html]

\end{thebibliography}
\end{document}